\documentclass[a4paper,11pt,dvipsnames]{elsarticle}
\usepackage[T1]{fontenc}
\usepackage[utf8]{inputenc}
\usepackage[english]{babel}
\usepackage{amsmath,amsthm,enumerate}
\usepackage{amssymb}
\usepackage{amscd}
\usepackage{vmargin}
\usepackage{url}
\usepackage{xcolor}
\usepackage{stmaryrd}
\usepackage{tikz,pgf}
\usepackage{xspace}
\usepackage{subfig}
\usepackage{cancel}
\usepackage{graphicx}
\usepackage[textsize=footnotesize]{todonotes}
\usetikzlibrary{positioning}
\usepackage{subfig}
\usepackage{adjustbox}
\usetikzlibrary{shapes.geometric}
\usepackage{soul}
\usepackage{esvect}
\usetikzlibrary{arrows}
\usepackage{MnSymbol}
\usetikzlibrary{decorations.pathreplacing,calligraphy}
\usepgflibrary{patterns.meta} 
\usepgflibrary[patterns.meta] 
\usetikzlibrary{patterns.meta} 
\usetikzlibrary[patterns.meta] 
\definecolor{myBlue}{HTML}{2D2F92}
\usetikzlibrary{decorations.markings}
\usepackage{array}
\usepackage{makecell}
\usepackage{tabularx}
\usepackage{colortbl}

\usetikzlibrary{decorations.markings}
\usetikzlibrary{decorations.pathmorphing}
\tikzset{snake it/.style={decorate, decoration={snake,amplitude=.6mm,segment length=2mm,post length=0mm}}}

\newcolumntype{Y}{>{\centering\arraybackslash}X}

\pgfdeclarepattern{name=hatch,
  parameters={\hatchsize,\hatchangle,\hatchlinewidth},
  bottom left={\pgfpoint{-.1pt}{-.1pt}},
  top right={\pgfpoint{\hatchsize+.1pt}{\hatchsize+.1pt}},
  tile size={\pgfpoint{\hatchsize}{\hatchsize}},
  tile transformation={\pgftransformrotate{\hatchangle}},
  code={
    \pgfsetlinewidth{\hatchlinewidth}
    \pgfpathmoveto{\pgfpoint{-.1pt}{-.1pt}}
    \pgfpathlineto{\pgfpoint{\hatchsize+.1pt}{\hatchsize+.1pt}}
    \pgfpathmoveto{\pgfpoint{-.1pt}{\hatchsize+.1pt}}
    \pgfpathlineto{\pgfpoint{\hatchsize+.1pt}{-.1pt}}
    \pgfusepath{stroke}
  }
}
\tikzset{
  hatch size/.store in=\hatchsize,
  hatch angle/.store in=\hatchangle,
  hatch line width/.store in=\hatchlinewidth,
  hatch size=5pt,
  hatch angle=0pt,
  hatch line width=.5pt,
}

\makeatletter
\def\ps@pprintTitle{%
 \let\@oddhead\@empty
 \let\@evenhead\@empty
 \def\@oddfoot{}%
 \let\@evenfoot\@oddfoot}
\makeatother

\usetikzlibrary{decorations.markings}
\tikzset{->-/.style={decoration={
  markings,
  mark=at position .75 with {\arrow{>}}},postaction={decorate}}}

\newif\ifcomments\commentstrue
\ifcomments
  \usepackage{todonotes}
  \newcommand{\commenttext}[1]{ \begin{center} {\fbox{\begin{minipage}[h]{0.9 \linewidth}   {\textsf{ #1}} \end{minipage} }} \end{center}}
  \newcommand{\Sam}[1]{{\color{red}\commenttext{Sam : #1}}}
  \newcommand{\Rom}[1]{{\color{blue}\commenttext{Rom : #1}}}
  
  \newcommand{\sam}[1]{\todo[size=\tiny,color=red!20!white]{\textbf{Sam:} #1}}
  \newcommand{\rom}[1]{\todo[size=\tiny,color=blue!20!white]{\textbf{Rom:} #1}}
  
\else
  \newcommand{\Sam}[1]{}
  \newcommand{\sam}[1]{}
  \newcommand{\Rom}[1]{}
  \newcommand{\rom}[1]{}
\fi
  
\begin{document}
\begin{frontmatter}

\title{Making Graphs Irregular through Irregularising Walks}

\author[nice]{Julien~Bensmail}
\author[bordeaux]{Romain~Bourneuf}
\author[grenoble]{Paul~Colinot}
\author[lyon]{Samuel~Humeau}
\author[orleans]{Timothée~Martinod}

\address[nice]{Universit\'e C\^ote d'Azur, CNRS, Inria, I3S, France}
\address[bordeaux]{Univ. Bordeaux, CNRS,  Bordeaux INP, LaBRI, UMR 5800, F-33400, Talence, France}
\address[grenoble]{Univ. Grenoble-Aples, CNRS, Grenoble INP, G-SCOP, UMR 5272, France}
\address[lyon]{ENS de Lyon, CNRS, Université Claude Bernard Lyon 1, LIP, UMR 5668, 69342, Lyon cedex 07, France}
\address[orleans]{Universit\'e d’Orl\'eans, INSA CVL, LIFO, UR 4022, Orl\'eans, France}

\begin{abstract}
The 1-2-3 Conjecture, introduced by Karo\'nski, {\L}uczak, and Thomason in 2004,
was recently solved by Keusch. This implies that, for any connected graph $G$ different from $K_2$,
we can turn $G$ into a locally irregular multigraph $M(G)$, \textit{i.e.}, in which no two adjacent vertices have the same degree,
by replacing some of its edges with at most three parallel edges.
In this work, we introduce and study a restriction of this problem under the additional constraint that edges added to $G$ to reach $M(G)$ must form a walk (\textit{i.e.}, a path with possibly repeated edges and vertices) of $G$.
We investigate the general consequences of having this additional constraint,
and provide several results of different natures (structural, combinatorial, algorithmic) on the length of the shortest irregularising walks,
for general graphs and more restricted classes.
\end{abstract}

 
\end{frontmatter}

\newtheorem{theorem}{Theorem}[section]
\newtheorem{lemma}[theorem]{Lemma}
\newtheorem{conjecture}[theorem]{Conjecture}
\newtheorem{observation}[theorem]{Observation}
\newtheorem{claim}[theorem]{Claim}
\newtheorem*{remark}{Remark}
\newtheorem{corollary}[theorem]{Corollary}
\newtheorem{proposition}[theorem]{Proposition}
\newtheorem{question}[theorem]{Question}
\newtheorem{definition}[theorem]{Definition}

\newcommand{\qedclaim}{\hfill $\diamond$ \medskip}
\newenvironment{proofclaim}{\noindent{\em Proof of the claim.}}{\qedclaim}

\newcommand{\p}{\textsf{P}\xspace}
\newcommand{\np}{\textsf{NP}\xspace}
\newcommand{\chis}{\chi'_{\rm s}}
\newcommand{\mlw}{{\rm ML}^{\rm W}}
\newcommand{\mlp}{{\rm ML}^{\rm P}}
\newcommand{\mvw}{{\rm MV}^{\rm W}}
\newcommand{\mvp}{{\rm MV}^{\rm P}}
\newcommand{\mew}{{\rm ME}^{\rm W}}
\newcommand{\mep}{{\rm ME}^{\rm P}}


\section{Introduction}

This work relates to investigations on the so-called 1-2-3 Conjecture,
which we first recall.
Let $G$ be a graph\footnote{Throughout, the term \textit{graph} always refers to a simple graph.}.
For any $k \geq 1$, a \textit{$k$-labelling} $\ell: E(G) \rightarrow \{1,\dots,k\}$ of $G$ is an assignment of \textit{labels}
(from $\{1,\dots,k\}$) to the edges of $G$.
From $\ell$, one can, for every vertex $u$ of $G$, compute the \textit{sum} $\sigma_\ell(u)$ of labels assigned to the edges incident to $u$,
that is, $\sigma_\ell(u)=\sum_{v \in N(u)} \ell(uv)$.
In case $\sigma_\ell$ is a proper vertex-colouring of $G$, that is,
if we have $\sigma_\ell(u) \neq \sigma_\ell(v)$ for every two adjacent vertices $u$ and $v$ of $G$,
we say that $\ell$ is \textit{proper}.
Now, it might be that $G$ does not admit any proper labellings at all.
However, it is known (see \textit{e.g.}~\cite{KLT04}) that the only connected graph with no proper labellings is $K_2$ (the complete graph of order~$2$).
We thus say $G$ is \textit{nice} whenever it has no connected components isomorphic to $K_2$,
and we define $\chis(G)$ as the smallest $k \geq 1$ such that $G$ admits proper $k$-labellings.

These notions were first introduced by Karo\'nski, {\L}uczak, and Thomason in~\cite{KLT04}.
The main focus at the time was to establish constant bounds on $\chis(G)$ that would hold whatever $G$ is.
The 1-2-3 Conjecture they raised in 2004 postulated that we should have $\chis(G) \leq 3$ for all nice graphs $G$,
which would be tight (see \textit{e.g.}~$K_3$ for an easy example, and odd-length cycles or complete graphs for less obvious ones). 
This conjecture generated a lot of appealing works on the topic,
eventually resulting in Keusch proving the 1-2-3 Conjecture in 2024~\cite{Keu24}.
For the sake of keeping this introduction legible, we will voluntarily not elaborate too much on the many other investigations that led to this milestone result,
and instead refer the interested reader to recent papers and surveys~\cite{Gal98,Sea12} on the topic for more information 
on the quest towards proving the 1-2-3 Conjecture.

\medskip

When it comes to discussing proper labellings and the 1-2-3 Conjecture,
there is a troublesome, recurring topic, 
which deals with the ``why'' of investigating these and focusing on the parameter $\chis$.
The paper that introduced the 1-2-3 Conjecture~\cite{KLT04} does not provide any particular motivations whatsoever.
Some can be retrieved from~\cite{CJLORS88}, in which the authors introduced and studied the notion of irregularity strength of graphs,
of which the parameter $\chis$ stands as a local variant.
Indeed, in the notions considered in~\cite{CJLORS88}, the goal is still to design labellings minimising some parameter,
but the main difference is that all vertices, not just adjacent ones, are required to have distinct sums.
There, however, the authors discuss motivations, which, for simplicity,
we translate to proper labellings and the parameter $\chis$.
Given a graph $G$ and a proper labelling $\ell$,
note that if, in $G$, we replace every edge $e \in E(G)$ with $\ell(e)$ parallel edges,
then we obtain a multigraph $M(G)$ with the same adjacencies as in $G$ (in the sense that any two vertices of $M(G)$ are adjacent if and only if they are in $G$)
and that is \textit{locally irregular}, \textit{i.e.}, with no two adjacent vertices having the same degree.
So, the fact that the 1-2-3 Conjecture is true implies that, whenever possible,
any graph can be transformed into a locally irregular multigraph by replacing edges with at most three parallel edges,
which, according to~\cite{CJLORS88}, might be considered preferable, 
as replacing edges with multiple edges can be regarded as more and more expensive as the number of replacing parallel edges increases.
This apart, we are not aware of any ``real'' applications supporting proper labellings and the 1-2-3 Conjecture. 

\begin{figure}[t]
\centering

\subfloat[A graph $G$]{
    \scalebox{0.65}{
    \begin{tikzpicture}[inner sep=0.7mm,thick]	

    \node[draw,circle,black,fill=white,line width=1pt](u1) at (1,0){\tiny $1$};
    \node[draw,circle,black,fill=white,line width=1pt](u2) at (2,0){\tiny $2$};
    \node[draw,circle,black,fill=white,line width=1pt](u3) at (4,0){\tiny $2$};
    \node[draw,circle,black,fill=white,line width=1pt](u4) at (6,0){\tiny $2$};
    \node[draw,circle,black,fill=white,line width=1pt](u5) at (8,0){\tiny $3$};
    \node[draw,circle,black,fill=white,line width=1pt](u6) at (8,1){\tiny $1$};
    \node[draw,circle,black,fill=white,line width=1pt](u7) at (10,0){\tiny $4$};
    \node[draw,circle,black,fill=white,line width=1pt](u8) at (10,1){\tiny $1$};
    \node[draw,circle,black,fill=white,line width=1pt](u9) at (11,0){\tiny $1$};
    \node[draw,circle,black,fill=white,line width=1pt](u10) at (10,-1){\tiny $4$};
    \node[draw,circle,black,fill=white,line width=1pt](u11) at (11,-1){\tiny $1$};
    \node[draw,circle,black,fill=white,line width=1pt](u12) at (10,-2){\tiny $1$};
    \node[draw,circle,black,fill=white,line width=1pt](u13) at (9,-1){\tiny $1$};
    
    \draw [-, line width=2pt,black] (u1) -- (u2);
    \draw [-, line width=2pt,black] (u2) -- (u3);
    \draw [-, line width=2pt,black] (u3) -- (u4);
    \draw [-, line width=2pt,black] (u4) -- (u5);
    \draw [-, line width=2pt,black] (u5) -- (u6);
    \draw [-, line width=2pt,black] (u5) -- (u7);
    \draw [-, line width=2pt,black] (u7) -- (u8);
    \draw [-, line width=2pt,black] (u7) -- (u9);
    \draw [-, line width=2pt,black] (u7) -- (u10);
    \draw [-, line width=2pt,black] (u10) -- (u11);
    \draw [-, line width=2pt,black] (u10) -- (u12);
    \draw [-, line width=2pt,black] (u10) -- (u13);
    
    \end{tikzpicture}
    }
}
\hspace{0pt}
\subfloat[A walk $W$ of $G$]{
    \scalebox{0.65}{
    \begin{tikzpicture}[inner sep=0.7mm,thick]	

    \node[draw,circle,black,fill=white,line width=1pt](u1) at (1,0){\tiny $1$};
    \node[draw,circle,black,fill=white,line width=1pt](u2) at (2,0){\tiny $2$};
    \node[draw,circle,black,fill=white,line width=1pt](u3) at (4,0){\tiny $2$};
    \node[draw,circle,black,fill=white,line width=1pt](u4) at (6,0){\tiny $2$};
    \node[draw,circle,black,fill=white,line width=1pt](u5) at (8,0){\tiny $3$};
    \node[draw,circle,black,fill=white,line width=1pt](u6) at (8,1){\tiny $1$};
    \node[draw,circle,black,fill=white,line width=1pt](u7) at (10,0){\tiny $4$};
    \node[draw,circle,black,fill=white,line width=1pt](u8) at (10,1){\tiny $1$};
    \node[draw,circle,black,fill=white,line width=1pt](u9) at (11,0){\tiny $1$};
    \node[draw,circle,black,fill=white,line width=1pt](u10) at (10,-1){\tiny $4$};
    \node[draw,circle,black,fill=white,line width=1pt](u11) at (11,-1){\tiny $1$};
    \node[draw,circle,black,fill=white,line width=1pt](u12) at (10,-2){\tiny $1$};
    \node[draw,circle,black,fill=white,line width=1pt](u13) at (9,-1){\tiny $1$};
    
    \draw [-, line width=2pt,black] (u1) -- (u2);
    \draw [-, line width=2pt,black] (u2) -- (u3);
    \draw [-, line width=2pt,black] (u3) -- (u4);
    \draw [-, line width=2pt,black] (u4) -- (u5);
    \draw [-, line width=2pt,black] (u5) -- (u6);
    \draw [-, line width=2pt,black] (u5) -- (u7);
    \draw [-, line width=2pt,black] (u7) -- (u8);
    \draw [-, line width=2pt,black] (u7) -- (u9);
    \draw [-, line width=2pt,black] (u7) -- (u10);
    \draw [-, line width=2pt,black] (u10) -- (u11);
    \draw [-, line width=2pt,black] (u10) -- (u12);
    \draw [-, line width=2pt,black] (u10) -- (u13);
    
    \draw [->-, line width=1pt,Red] (u2) to[bend left=20] node[midway,fill=white]{\fontsize{2.5}{4}\selectfont $1$} (u3);
    \draw [->-, line width=1pt,Red] (u3) to[bend left=20] node[midway,fill=white]{\fontsize{2.5}{4}\selectfont $2$} (u4);
    \draw [->-, line width=1pt,Red] (u4) to[bend right=40] node[midway,fill=white]{\fontsize{2.5}{4}\selectfont $3$} (u3);
    \draw [->-, line width=1pt,Red] (u3) to[bend left=60] node[midway,fill=white]{\fontsize{2.5}{4}\selectfont $4$} (u4);
    \draw [->-, line width=1pt,Red] (u4) to[bend left=20] node[midway,fill=white]{\fontsize{2.5}{4}\selectfont $5$} (u5);
    \draw [->-, line width=1pt,Red] (u5) to[bend right=40] node[midway,fill=white]{\fontsize{2.5}{4}\selectfont $6$} (u4);
    \draw [->-, line width=1pt,Red] (u4) to[bend left=60] node[midway,fill=white]{\fontsize{2.5}{4}\selectfont $7$} (u5);
    \draw [->-, line width=1pt,Red] (u5) to[bend left=20] node[midway,fill=white]{\fontsize{2.5}{4}\selectfont $8$} (u7);
    
    \end{tikzpicture}
    }
}

\subfloat[$G+W$ is locally irregular]{
    \scalebox{0.9}{
    \begin{tikzpicture}[inner sep=0.7mm,thick]	

    \node[draw,circle,black,fill=white,line width=1pt](u1) at (1,0){\tiny $1$};
    \node[draw,circle,black,fill=white,line width=1pt](u2) at (2,0){\tiny $3$};
    \node[draw,circle,black,fill=white,line width=1pt](u3) at (4,0){\tiny $6$};
    \node[draw,circle,black,fill=white,line width=1pt](u4) at (6,0){\tiny $8$};
    \node[draw,circle,black,fill=white,line width=1pt](u5) at (8,0){\tiny $7$};
    \node[draw,circle,black,fill=white,line width=1pt](u6) at (8,1){\tiny $1$};
    \node[draw,circle,black,fill=white,line width=1pt](u7) at (10,0){\tiny $5$};
    \node[draw,circle,black,fill=white,line width=1pt](u8) at (10,1){\tiny $1$};
    \node[draw,circle,black,fill=white,line width=1pt](u9) at (11,0){\tiny $1$};
    \node[draw,circle,black,fill=white,line width=1pt](u10) at (10,-1){\tiny $4$};
    \node[draw,circle,black,fill=white,line width=1pt](u11) at (11,-1){\tiny $1$};
    \node[draw,circle,black,fill=white,line width=1pt](u12) at (10,-2){\tiny $1$};
    \node[draw,circle,black,fill=white,line width=1pt](u13) at (9,-1){\tiny $1$};
    
    \draw [-, line width=2pt,black] (u1) -- (u2);
    \draw [-, line width=2pt,black] (u2) -- (u3);
    \draw [-, line width=2pt,black] (u3) -- (u4);
    \draw [-, line width=2pt,black] (u4) -- (u5);
    \draw [-, line width=2pt,black] (u5) -- (u6);
    \draw [-, line width=2pt,black] (u5) -- (u7);
    \draw [-, line width=2pt,black] (u7) -- (u8);
    \draw [-, line width=2pt,black] (u7) -- (u9);
    \draw [-, line width=2pt,black] (u7) -- (u10);
    \draw [-, line width=2pt,black] (u10) -- (u11);
    \draw [-, line width=2pt,black] (u10) -- (u12);
    \draw [-, line width=2pt,black] (u10) -- (u13);
    
    \draw [-, line width=2pt,black] (u2) to[bend left=20] (u3);
    \draw [-, line width=2pt,black] (u3) to[bend left=20] (u4);
    \draw [-, line width=2pt,black] (u4) to[bend right=40] (u3);
    \draw [-, line width=2pt,black] (u3) to[bend left=60] (u4);
    \draw [-, line width=2pt,black] (u4) to[bend left=20] (u5);
    \draw [-, line width=2pt,black] (u5) to[bend right=40] (u4);
    \draw [-, line width=2pt,black] (u4) to[bend left=60] (u5);
    \draw [-, line width=2pt,black] (u5) to[bend left=20] (u7);
    
    \end{tikzpicture}
    }
}
	
\caption{An example of an irregularising walk.
A graph $G$ is given (a), together with a walk $W$ (b).
Adding the edges of $W$ to $G$ results in $G+W$ (c), which is locally irregular.
Thus, $W$ is an irregularising walk of $G$.
Numbers in vertices indicate their degree.
In (b), the edges of $W$ are the integer-labelled arcs in red.
Numbers on arcs of $W$ allow to retrieve the order in which they are traversed.
}
\label{figure:example}
\end{figure}
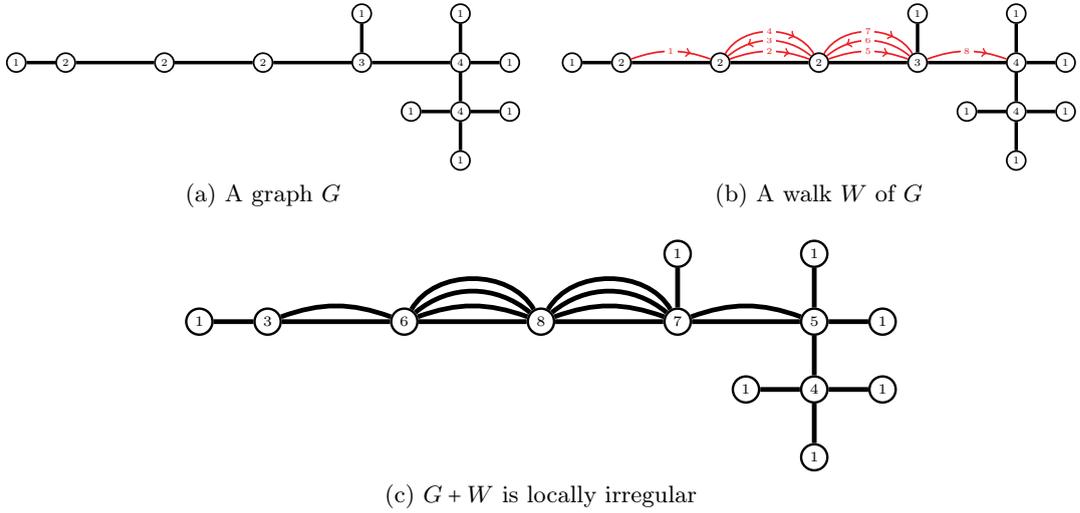

Thus, in a nutshell, through proper labellings and the 1-2-3 Conjecture, 
we are investigating ways to turn graphs into locally irregular multigraphs through increasing edge multiplicities,
and, this, in a somewhat optimal way.
And, in this convention, edges can be chosen and multiplied at will. 
From the point of view of hypothetical applications to these notions,
we believe this point is not quite realistic and too permissive.
We overcome this point in the new problem and notions we introduce by requiring multiplied edges to be obtained 
through navigating from edges to edges along a walk of the graph (see Figure~\ref{figure:example}).
To avoid ambiguities, recall that, in a graph $G$, a \textit{walk} is a sequence of vertices, every two subsequent of which are adjacent. For convenience, most of the time we will instead describe walks through their consecutive edges.
A walk of $G$ is said \textit{closed} if it starts and ends in the same vertex.
Note that the definition of walks allows edges and vertices to be repeated.
Recall further that a \textit{path} of $G$ is a walk with no repeated vertices, and that a \textit{cycle} of $G$ is a closed walk where all vertices are distinct except for the first and the last (which are the same). 
Now, regarding the earlier concerns, we say that a walk $W$ of $G$ is \textit{irregularising} if $G+W$ is locally irregular, where $G+W$ is the multigraph with vertex set $V(G)$ and edge multiset $E(G) \cup E(W)$ obtained when adding, preserving their multiplicities, the edges of $W$ to $G$. Let us insist on the fact that $G+W$ is indeed a multigraph, and that $E(G+W)$ is a multiset.

So, if $G$ is a graph and $W$ is an irregularising walk of $G$,
then $W$ corresponds to a proper labelling of $G$ with specific properties.
Obviously, graphs having two connected components being not locally irregular do not admit any irregularising walk; for simplicity, throughout, \textit{nice graphs} are thus all assumed connected.
As will be seen, all nice
graphs admit irregularising walks.
But, again, it seems natural to include some optimisation concerns to the problem,
and one can come up with many possible parameters of interest to consider.
For the sake of keeping this introduction short,
we elaborate on this point in the next section.
\medskip

This work is organised as follows.
We begin in Section~\ref{section:def} by pointing out several flexible parameters and aspects of interest behind irregularising walks,
and, from this, settle the exact parameters to be investigated in the next sections. In Section~\ref{section:early-properties}, we then raise first easy properties of these parameters of interest. In particular, from there on, we mostly focus on determining the length of shortest irregularising walks. In Section~\ref{section:bounds}, we provide general upper bounds on this parameter, expressed mainly in terms of graph size (number of edges),
chromatic number, and maximum degree. Then, in Section~\ref{section:particular-classes}, we focus on more particular classes of graphs, namely complete graphs, complete bipartite graphs, paths, and cycles. For each of these classes of graphs, we provide tight results on the length of shortest irregularising walks. Last, in Section~\ref{section:algo}, we focus on algorithmic aspects, proving mostly, on the negative side, that determining the length of shortest irregularising walks is \np-complete in general; and, on the positive side, that this can be achieved in polynomial time for trees. We finish off in Section~\ref{section:ccl} with perspectives and directions for further work on the topic.

\section{Formal definitions and parameters}\label{section:def}

As mentioned in the previous section, the notion of irregularising walks is rather flexible,
and, w.r.t.~previous investigations on proper labellings and the 1-2-3 Conjecture,
there are several interesting parameters that one could investigate.
In the current section, we discuss these open points, establish connections with previous works of the literature,
and explain why, along the next sections, we stick to the conventions and parameters we chose.

There are, essentially, two main points we need to discuss.
On the one hand, we discuss the types of walks we allow when designing irregularising walks:
closed walks, paths, restricted walks, etc.~On the other hand, we come up with several objectives of interest related to irregularising walks that one could consider:
minimising the total walk length, the maximum number of times vertices or edges are traversed, etc.

\subsection{Types of walks}\label{subsection:walk-types}

As suggested in the introductory section, we will mostly focus, in our investigations on irregularising walks, 
on walks in the most general setting where repetitions are allowed and the end-vertices can coincide.
But one could also focus on more restricted contexts.
In particular, paths (where the only possible repetition is on the end-vertices) are a fundamental notion of graph theory,
and it thus makes sense to also wonder about irregularising paths.
As will be mentioned later on, this is also an aspect we will keep in mind throughout.

There are also other restrictions on walks we could impose.
For instance, we could specifically focus on irregularising closed walks and irregularising cycles.
W.r.t.~aspects to be mentioned in the next subsection, we could also impose, in the case of walks,
that edges and/or vertices are not traversed too many times.
Another interesting aspect that one could consider, is limiting the number of \textit{half-turns}, 
\textit{i.e.}, pairs of consecutive edges of the form $(uv,vu)$, by irregularising walks, 
which is justified by the fact that several of our later proofs rely heavily on being able to perform these.

\subsection{Parameters of interest}

\subsubsection{Connections with previous works}

Recall that all our investigations on irregularising walks are initially inspired by the problem of turning, through edge multiplications, 
a simple graph into a locally irregular multigraph. And this latter matter is equivalent to designing proper labellings of graphs.
As mentioned implicitly earlier, all nice graphs admit proper labellings;
due to this fact, it seems legitimate to wonder about somewhat ``optimised'' proper labellings,
translating into nice properties of the associated locally irregular multigraphs.
Among others, to our knowledge the most studied aspects include:

\begin{itemize}
	\item Minimising the maximum assigned label; 
	this is equivalent to minimising the maximum number of parallel edges joining any two vertices in the associated multigraph --
	and this is the main concern behind the parameter $\chis$ and the 1-2-3 Conjecture~\cite{KLT04}.
	
	\item Minimising the sum of assigned labels;
	this is equivalent to minimising the size (number of edges) of the associated multigraph --
	this concern was considered in~\cite{BFN22}.
	
	\item Minimising the maximum sum of a vertex;
	this is equivalent to minimising the maximum degree of the associated multigraph --
	this question was investigated in~\cite{BLLN21}.
\end{itemize}

\noindent It is important to mention, see the references above,
that these three main matters are different. Of course, getting results on one of them might have consequences for the others.
For instance, since the 1-2-3 Conjecture holds~\cite{Keu24}, every nice graph $G$ can, through edge multiplications,
be turned into a locally irregular multigraph with size at most $3|E(G)|$ and maximum degree at most $3\Delta(G)$.
However, the authors of the references above also came up with constructions of nice graphs showing that minimising any of these aspects above
is arbitrarily far from minimising the others.

It would be a bit too demanding to the readers to survey now all results known on the three problems above.
For a better flow of what is to be exposed,
we will deliver important points throughout this work, as they are needed to get a good grasp on our results.

\medskip

In the very same vein, from the associated multigraph point of view, 
we believe that several parameters of irregularising walks could be worth optimising.
Namely:

\begin{itemize}
	\item To us, the most natural parameter to optimise is the length (\textit{i.e.}, number of edges)
	of the irregularising walk. From the multigraph point of view,
	this is equivalent to minimising the size of the associated multigraph.
	
	\item Still w.r.t.~walks, another interesting parameter to optimise could be the maximum number of times some edge is traversed.
	From the multigraph point of view, this is equivalent to minimising the maximum number of parallel edges joining two vertices.
	
	\item W.r.t.~the previous concerns, one could also wonder about optimising the maximum number of times a vertex is traversed.
	From the multigraph point of view, this sort of relates to minimising the maximum degree of a vertex.
\end{itemize} 

\noindent Of course, one can certainly come up with other parameters of interest.
For now, however, we will focus on these three, throughout this work.

\subsubsection{Formal definitions}

As suggested in the previous two subsections, the problem we are introducing has many flexible points,
and, for the sake of making our investigations clear, it is important that we clarify the exact conventions and properties
we will focus on from now on.

This work being the very first one on the topic,
it seems more appropriate to focus on irregularising walks and paths with no further path restrictions.
In particular, we will allow ourselves to design irregularising closed walks and cycles.
More precisely, we will mostly focus on irregularising walks, for the simple reason that all nice graphs admit irregularising walks (Corollary~\ref{corollary:bound-general})
while many nice graphs do not admit irregularising paths at all. Even worse, determining whether a nice graph admits irregularising paths is \np-complete (Theorem~\ref{theorem:npc-paths}),
which implies that, unless $\p = \np$, there is no ``simple'' characterisation of graphs admitting irregularising paths.

In terms of parameters, for the reasons mentioned earlier, we will mostly focus on the parameter $\mlw(G)$ (``\underline{M}inimum \underline{L}ength; \underline{W}alks''),
which, for a graph $G$, is the smallest $k \geq 1$ such that $G$ admits irregularising walks of length~$k$.
At times, we will also discuss the two parameters $\mew(G)$ (``\underline{M}inimum \underline{E}dge; \underline{W}alks'') and
$\mvw(G)$ (``\underline{M}inimum \underline{V}ertex; \underline{W}alks'') which refer to the smallest $k \geq 1$ such that $G$ admits irregularising walks
where edges and vertices, respectively, are traversed at most $k$ times each (for $\mvw(G)$, it should be understood that we aim at minimising the maximum number of edges of $W$ incident to some vertex, 
by some irregularising walk $W$ of $G$).
We extend these three parameters to irregularising paths as $\mlp(G)$, $\mep(G)$, and $\mvp(G)$, respectively,
in the obvious way.
Naturally, any of these six parameters is set to $+\infty$ in case $G$ does not admit any corresponding irregularising walk or path (which is the case if $G$ is not nice).

\section{Early properties, and relationships between parameters}\label{section:early-properties}

Now that our new notions of interest have been properly introduced, we begin in this section by investigating some first aspects. We start by stating formal relationships with related distinguishing labelling problems,
mentioned in the introductory section. We then establish general connections between all of the three parameters $\mlw$, $\mew$, and $\mvw$.

\subsection{Relationships with labelling parameters}\label{subsection:relation-labellings}

We mentioned in the previous section that designing irregularising walks minimising the walk length, the maximum number of times an edge is traversed,
and the maximum number of times a vertex is traversed, is related to the problem of designing proper labellings with minimum label sum~\cite{BFN22},
minimum largest assigned label~\cite{KLT04}, and minimum vertex sum~\cite{BLLN21}.
From this, we deduce the next result. Note, in particular, that the additive terms appear because, in labelled graphs, all edges are assigned labels with value at least~$1$,
while, in an irregularising walk, we count how many times edges are traversed (which can be $0$).

\begin{observation}\label{observation:connections-with-labellings}
Let $G$ be a nice graph of size $m$ and maximum degree $\Delta$.
\begin{itemize}
	\item If $x$ is the minimum label sum of a proper labelling of $G$,
	then $x \leq \mlw(G)+m$.
	\item We have $\chis(G) \leq \mew(G)+1$.
	\item If $x$ is the minimum vertex sum of a proper labelling of $G$,
	then $x \leq \mvw(G)+\Delta$.
\end{itemize}
\end{observation}

\begin{proof}
This follows from the fact that if claimed irregularising walks of $G$ exist,
then we directly obtain proper labellings of $G$ with the claimed properties.
\end{proof}

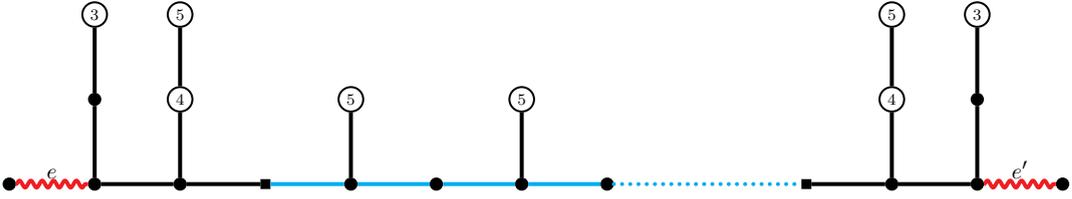
\begin{figure}[t]
\centering

\scalebox{0.75}{
    \begin{tikzpicture}[inner sep=0.7mm,thick]	
    \node[draw,circle,black,fill=black,line width=1pt](u1) at (0,0){};
    \node[draw,circle,black,fill=black,line width=1pt](u2) at (1.5,0){};
    \node[draw,circle,black,fill=black,line width=1pt](u3) at (3,0){};
    \node[draw,black,fill=black,line width=1pt](u4) at (4.5,0){};
    \node[draw,circle,black,fill=black,line width=1pt](u5) at (6,0){};
    \node[draw,circle,black,fill=black,line width=1pt](u6) at (7.5,0){};
    \node[draw,circle,black,fill=black,line width=1pt](u7) at (9,0){};
    \node[draw,circle,black,fill=black,line width=1pt](u8) at (10.5,0){};
    \node[draw,black,fill=black,line width=1pt](u9) at (14,0){};
    \node[draw,circle,black,fill=black,line width=1pt](u10) at (15.5,0){};
    \node[draw,circle,black,fill=black,line width=1pt](u11) at (17,0){};
    \node[draw,circle,black,fill=black,line width=1pt](u12) at (18.5,0){};
    
    \node[draw,circle,black,fill=black,line width=1pt](v2) at (1.5,1.5){};
    \node[draw,circle,black,fill=white,line width=1pt](v3) at (3,1.5){\scriptsize $4$};
    \node[draw,circle,black,fill=white,line width=1pt](v5) at (6,1.5){\scriptsize $5$};
    \node[draw,circle,black,fill=white,line width=1pt](v7) at (9,1.5){\scriptsize $5$};
    \node[draw,circle,black,fill=white,line width=1pt](v10) at (15.5,1.5){\scriptsize $4$};
    \node[draw,circle,black,fill=black,line width=1pt](v11) at (17,1.5){};
    
    \node[draw,circle,black,fill=white,line width=1pt](w2) at (1.5,3){\scriptsize $3$};
    \node[draw,circle,black,fill=white,line width=1pt](w3) at (3,3){\scriptsize $5$};
    \node[draw,circle,black,fill=white,line width=1pt](w10) at (15.5,3){\scriptsize $5$};
    \node[draw,circle,black,fill=white,line width=1pt](w11) at (17,3){\scriptsize $3$};
    
    \draw [-, line width=2pt,Red,snake it] (u1) -- (u2) node [midway, above, black]{$e$};
    \draw [-, line width=2pt,black] (u2) -- (u3);
    \draw [-, line width=2pt,black] (u3) -- (u4);
    \draw [-, line width=2pt,Cyan] (u4) -- (u5);
    \draw [-, line width=2pt,Cyan] (u5) -- (u6);
    \draw [-, line width=2pt,Cyan] (u6) -- (u7);
    \draw [-, line width=2pt,Cyan] (u7) -- (u8);
    \draw[-,line width=2pt,loosely dotted,line cap=round,dash pattern=on 0pt off 2\pgflinewidth,Cyan] (u8) -- (u9);
    \draw [-, line width=2pt,black] (u9) -- (u10);
    \draw [-, line width=2pt,Red,black] (u10) -- (u11);
    \draw [-, line width=2pt,Red,snake it] (u11) -- (u12) node [midway, above, black]{$e'$};;
    
    \draw [-, line width=2pt,black] (u2) -- (v2);
    \draw [-, line width=2pt,black] (u3) -- (v3);
    \draw [-, line width=2pt,black] (u5) -- (v5);
    \draw [-, line width=2pt,black] (u7) -- (v7);
    \draw [-, line width=2pt,black] (u10) -- (v10);
    \draw [-, line width=2pt,black] (u11) -- (v11);
    
    \draw [-, line width=2pt,black] (v2) -- (w2);
    \draw [-, line width=2pt,black] (v3) -- (w3);
    \draw [-, line width=2pt,black] (v10) -- (w10);
    \draw [-, line width=2pt,black] (v11) -- (w11);
    \end{tikzpicture}
}
	
\caption{In this graph $G$, vertices in black have their full neighbourhood displayed, while, for vertices in white, we only provide their degrees (through a number inside -- missing neighbours can be \textit{e.g.}~leaves). A proper labelling minimising the label sum of this graph is obtained by assigning label~$2$ to the two red (wiggly) edges $e$ and $e'$, and label~$1$ to all other edges. On the other hand, the unique path $P$ containing $e$ and $e'$ contains all blue edges (between the two square vertices), which would generate new conflicts in $G+P$ between original vertices of degree~$3$ and their neighbours of degree~$5$. Thus, any irregularising walk of $G$ traversing $e$ and $e'$ would have to contain extra edges.}
\label{figure:hard-to-connect}
\end{figure}
Establishing bounds that involve inequalities reversed from those in Observation~\ref{observation:connections-with-labellings} is not as easy.
Indeed, there are examples of graphs that are not locally irregular, in which degree conflicts between neighbours are ``far apart'',
and in which, through a labelling, getting rid of these conflicts can be done very easily (by \textit{e.g.}~assigning label~$2$ locally, to some edge).
And, in such a graph, given an optimal labelling $\ell$ (in terms of label sum, maximum assigned label, or maximum sum),
``converting'' $\ell$ to a ``good'' irregularising walk is not trivial.
To convince the reader of this fact, we provide an example in Figure~\ref{figure:hard-to-connect}. In this graph $G$, a proper labelling $\ell$ can be obtained by assigning label~$2$ to two edges $e$ and $e'$, and label~$1$ to the rest of the edges. In particular, the fact that label~$2$ is assigned only to $e$ and $e'$ guarantees that $\ell$ has convenient properties (very small label sum, maximum assigned label, and maximum sum). On the other hand, any irregularising walk $W$ of $G$ containing $e$ and $e'$ would have to traverse other edges, thereby increasing the degrees of some vertices and generating new conflicts in $G+W$. So $W$ has to go through additional edges to get rid of these, and the situation could get even worse if $G$ had other structures.

Still, combining an approach and arguments to be introduced in later Theorem~\ref{theorem:bound-general}, we will be able to establish inequalities reverse to those in Observation~\ref{observation:connections-with-labellings}; we postpone these to later Theorem~\ref{theorem:bounds-parameters-greedy-proper}.

To get a better grasp of what Observation~\ref{observation:connections-with-labellings} actually means, 
let us recall that, because the 1-2-3 Conjecture holds~\cite{Keu24},
for every nice graph $G$, over all proper labellings, the minimum label sum is at most $3|E(G)|$,
the minimum largest assigned label is at most~$3$,
and the minimum vertex sum is at most $3\Delta(G)$.
While only the bound on the second parameter is known to be best possible,
the first one is conjectured to be at most about $2|E(G)|$ (see~\cite{BFN22}),
while things are less clear regarding the third one (see~\cite{BLLN21}).

\subsection{Relationships between the irregularising walk parameters}

Paths being restricted types of walks, we immediately get:

\begin{observation}\label{observation:paths-bounds-walks}
For every graph $G$,
we have $\mlw(G) \leq \mlp(G)$, $\mew(G) \leq \mep(G)$, and $\mvw(G) \leq \mvp(G)$.
\end{observation}

\noindent This might indicate that, generally speaking, one could investigate irregularising paths in order to derive results on irregularising walks.
As mentioned earlier, unfortunately not all nice graphs admit irregularising paths (while they all admit irregularising walks; see later Corollary~\ref{corollary:bound-general}),
and actually describing graphs with no irregularising paths is not easy (see later Theorem~\ref{theorem:npc-paths}).
So, there is not much we can get from Observation~\ref{observation:paths-bounds-walks}.

Since the three parameters $\mlw$, $\mew$, and $\mvw$ all deal with irregularising walks, we can derive some straightforward relationships involving them. We focus on irregularising walks in what follows, but we could easily adapt this result for paths.

\begin{observation}\label{observation:relationships}
If $G$ is a nice graph of order $n$, size $m$, and maximum degree $\Delta$, then:
\begin{itemize}
    \item $\mew(G) \leq \mlw(G)$ and $\mvw(G) \leq \mlw(G)$;
    \item $\mlw(G) \leq m \cdot \mew(G)$ and $\mvw(G) \leq \Delta \cdot \mew(G)$;
    \item $\mlw(G) \leq \frac{n}{2} \cdot \mvw(G)$ and $\mew(G) \leq \mvw(G)$.
\end{itemize}
In particular: 
\begin{itemize}
    \item $\mew(G) \leq \mlw(G) \leq m \cdot \mew(G)$ and $\mvw(G) \leq \mlw(G) \leq \frac{n}{2} \cdot \mvw(G)$;
    
    \item $\frac{1}{m} \cdot \mlw(G) \leq \mew(G) \leq \mlw(G)$ and $\frac{1}{\Delta} \cdot \mvw(G) \leq \mew(G) \leq \mvw(G)$;
    
    \item $\frac{2}{n} \cdot \mlw(G) \leq \mvw(G) \leq \mlw(G)$ and $\mew(G) \leq \mvw(G) \leq \Delta \cdot \mew(G)$.
\end{itemize}
\end{observation}

\begin{proof}
Let $W$ be an irregularising walk of $G$.

\begin{itemize}
    \item Assume that $W$ has length $x$. Then each edge of $G$ is traversed at most $x$ times by $W$. Likewise, for every vertex of $G$, its incident edges are traversed, in total, at most $x$ times by $W$. This proves the first item.
    
    \item Assume that every edge of $G$ is traversed at most $x$ times by $W$. Then the length of $W$ is at most $mx$. Likewise, for every vertex $v$ of $G$, each of its $d(v)$ incident edges is traversed at most $x$ times by $W$, and thus $v$ is traversed at most $x d(v) \leq x \Delta$ times (in the sense that there are at most $xd(v)$ edges of $W$ incident to $v$). This proves the second item.
    
    \item Assume that every vertex of $G$ is traversed at most $x$ times by $W$, \textit{i.e.}, there are at most $x$ edges of $W$ incident to $v$. The length of $W$ is thus at most $\frac{n}{2}x$ since edges are counted twice. Likewise, if an edge $uv$ of $G$ was traversed at least $x+1$ times, then $u$ (and $v$) would be incident to $x+1$ edges of $W$, a contradiction. Thus, every edge is traversed at most $x$ times. This proves the third item.
\end{itemize}

The last three items follow from the first three ones.
\end{proof}


\section{General bounds}\label{section:bounds}

In this section, we provide general bounds on the parameter $\mlw$ expressed in terms of other graph parameters. The first one we establish, through Corollary~\ref{corollary:bound-general}, is a function of the graph's order and size, while in the second one we provide, in Corollary~\ref{corollary:bound-chromatic}, one of the main parameters is the graph's chromatic number.

\subsection{General graphs}

We start off with the main arguments behind both bounds. For a walk $W$, we denote by $|W|$ its number of vertices and by $\|W\|$ its number of edges.

\begin{theorem}\label{theorem:bound-general}
Let $G$ be a nice connected graph of size $m$,
and $W$ be a closed walk of length~$p$ of $G$ going through all vertices. 
Then,$$\mlw(G) \leq p+2m.$$
\end{theorem}

\begin{proof}
Set $W=u_0 \dots u_{p-1}u_p$, where $d(u_0) > 1$ (since $G$ is nice and connected, we have $\Delta(G)>1$),
and $u_0=u_p$.
Let also $v_1$ and $v_2$ be two distinct neighbours of $u_0$, where, say, $v_1=u_{p-1}$.
We build an irregularising walk $W'$ of $G$ by essentially ``following'' $W$, and, along the way, performing, when needed, sufficiently many half-turns to get rid of conflicts in the eventual multigraph $G+W'$. Recall that we mentioned the notion of half-turns in Section~\ref{subsection:walk-types}; for reminder, a \textit{half-turn} in a walk consists, being at some vertex $u$ at some point, in going to a neighbour $v$, and coming back to $u$ (thus, the walk contains $uvu$ as a subwalk). Most of the time, these half-turns will be performed along  edges going to subsequent vertices of $W$; that is, when considering any $u_i$, we will go to $u_{i+1}$ following $W$, and then perform half-turns of the form $u_{i+1}u_iu_{i+1}$.
A downside is that, in doing so, we alter the degree of $u_{i+1}$ in $G+W'$; however, $u_{i+1}$ is to be considered next, so we will have the opportunity to further modify its degree, if needed.

A crucial point in the proof is that if some vertex $u$ of $G$ is traversed several times by $W$, \textit{i.e.}, for some $i<j$ we have $u=u_i=u_j$, then, when building $W'$, making sure that $u$ is distinguished from (some of) its neighbours by performing half-turns (around $u$) should only be performed the last time we encounter $u$ when following $W$. This is because, due to the half-turns we perform (along the edge to the next vertex), this is the last opportunity to do so, and also because later half-turns will not impact what the degree of $u$ will be in the eventual $G+W'$. 
Another technical point is that, when reaching $u_p$, there is no further edge (of the form $u_pu_{p+1}$) to work with to make sure that $u_p$ is not involved in conflicts. The solution we come up with consists, when performing half-turns along $u_{p-1}u_p$ (to treat $u_{p-1}$), in performing additional half-turns, if needed, to guarantee that $u_p$ and $v_2$ (recall that $u_{p-1}=v_1 \neq v_2$) can never be in conflict in $G+W'$, even if additional half-turns along $u_pv_2$ are later performed. Eventually, to be done, it will then suffice to perform, if needed, additional half-turns along $u_pv_2$ to get rid of any conflicts involving $u_p$ or $v_2$.

Let us now describe how to obtain $W'$.
Starting from $u_0$, we follow the vertices of $W$ one by one in order.
For every vertex $u_i$ considered this way, we proceed as follows.	

\begin{itemize}
	\item If $u_i$ is to be reconsidered later in the process
	(\textit{i.e.}, there is some $j>i$ such that $u_j=u_i$; or $u_i=v_2$ for reasons explained above),
	then we just continue to $u_{i+1}$ (that is, we just add $u_iu_{i+1}$ to $W'$).
	
	\item If $i < p-1$ and $u_i$ is not going to be reconsidered later in the process 
	(\textit{i.e.}, there is no $j>i$ such that $u_j=u_i$; and $u_i\neq v_2$ again),
	then we go to $u_{i+1}$ (we add $u_iu_{i+1}$ to $W'$) and, before continuing with $u_{i+1}$,
	we first perform sufficiently many half-turns (added to $W'$) from $u_{i+1}$ to $u_i$
	so that, in the resulting $W'$, the degree of $u_i$ in $G+W'$ is not the same
	as that of the neighbours of $u_i$ that are not to be reconsidered later in the process
	(that is, neighbours of $u_i$ that will not see their degrees evolve further in $G+W'$).
	
	\item If $i=p-1$, then we essentially perform the same, but we additionally make sure that the resulting degree of $u_{i+1}=u_{p}$ in $G+W'$ is not the same as that of $v_2$
	(recall $v_2$ is a neighbour of $u_0=u_p$ different from $u_{p-1}=v_1$).
	
	\item If $i=p$, then we perform, in $W'$, sufficiently many half-turns from $u_p$ to $v_2$ so that the resulting degree of $u_p$ in $G+W'$ is not the same as the degrees of its neighbours different from $v_2$,
	and similarly the resulting degree of $v_2$ in $G+W'$ is not the same as the degrees of its neighbours different from $u_p$.
	In particular, since, prior to performing these half-turns, the degree in $G+W'$ of $u_p$ is not the same as that of $v_2$, note that the degrees of $u_p$ and $v_2$ cannot get equal upon performing
	half-turns from $u_p$ to $v_2$ (as every performed half-turn increases both degrees by~$2$).
\end{itemize}


All in all, compared to the original walk $W$, the length of $W'$ is increased only through half-turns. For a final analysis, note that whenever we perform a half-turn, we do it because of some edge $uv$ whose incident vertices are in conflict, with two possible cases. Either $u=u_p$ and $v=v_2$, in which case the conflict is handled by adding one half-turn along $u_{p-1}u_p$, or, w.l.o.g, $v=u_i$ and $u$ has already been treated and is not to be reconsidered later in the process. In the latter case, the conflict is handled by adding one half-turn along $u_iu_{i+1}$. Edge $uv$ cannot raise several conflicts during the process, as either $u=u_p$ and $v=v_2$ and $W'$ will only traverse $uv$ after taking care of the conflict, or neither $u$ nor $v$ will be seen again on $W'$ once the conflict has been taken care of. Overall, every edge is thus responsible for at most one conflict, hence one half-turn. So, in total, we have $$\|W'\| \leq p+2m.$$ This concludes the proof.
\end{proof}

\begin{corollary}\label{corollary:bound-general}
If $G$ is a nice connected graph of order $n$ and size $m$,
then $$\mlw(G) \leq 2(m+n-1).$$
In particular, since $m \geq n-1$, $$\mlw(G) \leq 4m.$$
\end{corollary}

\begin{proof}
This follows from the fact that considering any spanning tree of $G$ and performing, in that tree, a DFS algorithm from any vertex yields a closed walk $W$ of $G$ of length $2(n-1)$, going through all vertices.
Theorem~\ref{theorem:bound-general} then yields the bound $\mlw(G)\leq 2(n-1)+2m$.
\end{proof}

The last part of Corollary~\ref{corollary:bound-general} makes particular sense to us, since the parameter $\mlw$ is a length parameter, thus dealing with a number of edges. In particular, connected graphs of size $m$ all admit closed walks of length at most about $2m$, going through all vertices,
and Corollary~\ref{corollary:bound-general} says that irregularising walks can be achieved through walks that are, roughly, only at most twice as long.

\medskip

It is worth noting that the proof of Theorem~\ref{theorem:bound-general} is definitely improvable in some contexts, so Corollary~\ref{corollary:bound-general} is not best possible in general. In particular, in Theorem~\ref{theorem:bound-general}, we require that $W$ is a closed walk to guarantee that we treat all vertices of $G$ and get the opportunity to alter their degrees (and get rid of conflicts) in $G+W'$. Clearly, this is too strong, as this condition on $W$ does not take into account that parts of $G$, originally, might already be locally irregular (and thus, following our reasoning, there is no need for $W$ to traverse vertices there). 
In particular, our proof is still valid if $W$ visits at least one of $u$ and $v$ for every two adjacent vertices $u$ and $v$ with $d_G(u)=d_G(v)$. This would also improve our bounds in the following way: if $\partial(W)$ denotes the set of edges joining vertices in $V(W)$ and $V(G)\backslash V(W)$, then our bound improves to $p+2(|E(G[V(W)])|+|\partial(W)|)$. 
An extreme case is obviously when $G$ is locally irregular, in which case these arguments apply for $W$ being the empty closed walk. Through later results, we will get to see other examples where such arguments apply.

Another possible room for improvement lies in the original walk $W$ which guides the construction of $W'$ in the proof of Theorem~\ref{theorem:bound-general}. First, as mentioned previously, it is not always mandatory that $W$ goes through all vertices. Also, to derive Corollary~\ref{corollary:bound-general}, we made use of the general fact that any connected graph of order $n$ admits a closed walk $W$ of length $2(n-1)$ going through all vertices. However, in some graphs, shorter closed walks with the same properties exist. For instance, in the case of a Hamiltonian graph, we could instead take such a $W$ of length~$n$ only (any Hamiltonian cycle).
Last, let us also remark that the last half-turns of the proof, performed along $u_pv_2$, could be simplified to a walk of any length along $u_pv_2$ (getting rid of conflicts involving $u_p$ and $v_2$). This would definitely lead to a negligible improvement over the resulting bound.

\medskip

In terms of tightness, the bound in Corollary~\ref{corollary:bound-general} is not too far from being tight, as, as will be pointed out later on (in Section~\ref{section:ccl}), there are graphs $G$ of size $m$ for which $\mlw(G)$ is about $3m$.

\subsection{Combining guiding walks and proper labellings}

We saw through Observation~\ref{observation:connections-with-labellings} that we can derive lower bounds on the parameters $\mlw$, $\mew$, and $\mvw$ from the existence of proper labellings with particular properties. It turns out that we can also derive upper bounds on these by combining this idea with the greedy procedure introduced in the proof of Theorem~\ref{theorem:bound-general}.

\begin{theorem}\label{theorem:bounds-parameters-greedy-proper}
Let $G$ be a nice graph of size $m$ and maximum degree $\Delta$.
\begin{itemize}
	\item If $x$ is the minimum label sum of a proper labelling of $G$,
	then $x \leq \mlw(G)+m \leq 3(m+\Delta x)$.
	\item We have $\chis(G) \leq \mew(G)+1 \leq 3(1+\Delta\chis(G))$.
	\item If $x$ is the minimum vertex sum of a proper labelling of $G$,
	then $x \leq \mvw(G)+\Delta \leq 3\Delta(1+x)$.
\end{itemize}
\end{theorem}

\begin{proof}
Since the lower bounds follow from Observation~\ref{observation:connections-with-labellings}, we focus on proving the upper bounds. Our ideas here mainly consist in following a guiding walk of $G$ (as in the proof of Theorem~\ref{theorem:bound-general}) to perform half-turns, the main difference here being that these half-turns are performed w.r.t.~a proper labelling of $G$.

In what follows, we need a particular closed walk $W$ of $G$, which can e.g.~be obtained as follows. Recall first that, since $G$ is nice, then $G$ is connected. Let $G'$ be the multigraph obtained from $G$ by replacing every edge by two parallel edges (or, in other words, by a cycle of length $2$). Then $G'$ is connected and every vertex of $G'$ has even degree, so, by Euler's Theorem, $G'$ has an Eulerian tour $W'$. Let now $W$ be the walk of $G$ obtained by mapping each edge in $W'$ to its corresponding edge in $G$. Then, $W$ is a closed walk of $G$ which traverses every edge exactly twice.



Now let $\ell$ be any proper labelling of $G$, and set $q=\lfloor 3\Delta/2 \rfloor$. Consider $W'$, the walk of $G$ obtained by following $W$ and, whenever traversing an edge $uv$ for the last (second) time (say, going from $u$ to $v$), performing $q \cdot (\ell(uv)-1)$ half-turns of the form $vuv$. 

For every vertex $u$ of $G$, set $c_u=\sigma_\ell(u)-d_G(u)$ (where, recall, $\sigma_\ell(u)$ is the sum of labels assigned by $\ell$ to edges incident to $u$; and thus $c_u$ is the number of new edges incident to $u$ in the associated multigraph, that is, not taking into account the original edges incident to $u$ in $G$).
Since $\ell$ is proper, for every edge $uv$ of $G$ we have $d_G(u)+c_u \neq d_G(v)+c_v$. We claim that $u$ and $v$ also have distinct degrees in $G+W'$; that is $$d_{G+W'}(u)=3d_G(u)+2qc_u \neq 3d_G(v)+2qc_v=d_{G+W'}(v).$$ To see this is true, consider two cases:
\begin{itemize}
    \item If $d_G(u) = d_G(v)$, then the inequality holds since $2qc_u \neq 2qc_v$ (since $c_u \neq c_v$).
    
    \item If $d_G(u) \neq d_G(v)$, then, for the inequality to hold, we need $$3(d_G(u)-d_G(v))+2q(c_u-c_v) \neq 0.$$ Clearly, this holds when $c_u = c_v$ since $d_G(u) \neq d_G(v)$. 
    Now, when $c_u \neq c_v$, we have $|2q(c_u-c_v)|\geq 2q > 3\Delta-2$, and $|3(d_G(u)-d_G(v))|\leq 3(\Delta-1) < 3\Delta - 2$, so the inequality above holds.
\end{itemize}

\noindent Thus, $W'$ is irregularising. 
The claimed upper bounds now follow by applying the arguments above with $\ell$ being any proper labelling with the corresponding property.
\end{proof}

Note that we can actually deduce slightly better upper bounds from the proof of Theorem~\ref{theorem:bounds-parameters-greedy-proper}; however, the stated bounds are better looking, which is why we keep them as is.

\subsection{On the shape of general walks in graphs}

Before continuing with upcoming Theorem~\ref{theorem:bound-chromatic}, which we prove in a similar way as Theorem~\ref{theorem:bound-general}, we provide Proposition~\ref{proposition:walk_shapes} below which somewhat justifies why it might be difficult to establish a general result better than Corollary~\ref{corollary:bound-general}.
More precisely, we show that, in any graph, all walks, through adding half-turns, build upon some walk in which edges appear at most twice. In other words, any walk can be obtained by adding half-turns along some walk traversing each edge at most twice.
To make our point more formal, we introduce some definitions.

Given a graph $G$ and a walk $W=u_0\dots u_l$ of $G$, another
walk $W'$ is said to \emph{follow} $W$ if there is a sequence $(i_0,\dots,i_{l-1})$ of non-negative integers such that $$W'=u_0(u_1u_0)^{i_0}\dots u_{l-1}(u_lu_{l-1})^{i_{l-1}}u_l$$ (where, when writing $e^i$ for some edge $e$ and $i \geq 0$, we mean that $i$ half-turns are performed along $e$).
Put differently, $W'$ can be obtained from $W$ by traversing $W$ from start to finish and adding half-turns along the way. Given two walks $W_1$ and $W_2$,
we say that $W_1$ and $W_2$ are \emph{equivalent} if their underlying edge multisets are equal. For instance, if we consider as $G$ the path $P_3=u_0u_1u_2u_3$ of length
$3$ and the walks $W_1=u_1u_0u_1u_2$, $W_2=u_2u_1u_0u_1$, and $W_3=u_0u_1u_2$
of $G$, then $W_1$ and $W_2$ are equivalent
but $W_3$ is equivalent to none of them. Note that a walk is irregularising if and only if all its equivalent walks are irregularising.

\begin{proposition}\label{proposition:walk_shapes}
  Let $G$ be a graph, and $W$ be a walk of $G$. Denote by $E_o$ and $E_e$ the sets of edges traversed by $W$
  an odd number of times, and a non-zero even number of times, respectively. Then, $W$ is equivalent to a walk $W_S$ that follows a walk $S$ in which each edge of $E_o$ appears once, and each edge of $E_e$ appears at most twice.
\end{proposition}
\begin{proof}
  We replace $W$ by an equivalent walk to gather occurrences of each edge.
  
  First, we show how this is done for individual edges.
  Let $e=xy$ be an edge traversed by $W$. Assume that $e$ occurs $n\ge 1$
  times in $W$ and let $W_1,\dots,W_{n+1}$ be the largest subwalks of $W$ not containing $e$; that is,
  \[
    W=W_1e_1W_2e_2\dots W_ne_nW_{n+1},
  \]
  with $e_i$ being oriented edges in $\{xy,yx\}$. Up to reversing them, the subwalks $W_i$ for $i\notin\{1,n+1\}$ can be partitioned into three sets:
  the ones in $\mathcal{W}_x$ going from $x$ to $x$, the ones in $\mathcal{W}_y$ going from $y$ to $y$, and those in $\mathcal{W}_{xy}$
  going from $x$ to $y$. W.l.o.g., suppose that $W_1$ goes from the first vertex of $W$ to $x$. The walk $W_{n+1}$ can start either from $x$ or from $y$, and then goes on to the last vertex of $W$. These two cases depend on the parity of $n+|\mathcal W_{xy}|$: if this number is even then $W_{n+1}$ must start from $x$, and if it is odd it starts from $y$. This fact will justify that equivalent walks below are well-defined.
  
  \begin{itemize}
      \item Assume first that $\mathcal{W}_{xy}$ is not empty. In this case, choose $W'\in\mathcal{W}_{xy}$, and replace
  $W$ by
  the equivalent walk
  \[
      W_1\left( \prod_{W_x\in\mathcal{W}_x}W_x \right)W'\left( \prod_{W_y\in\mathcal{W}_y}W_y \right) e^n\left(\prod_{W_{xy}\in\mathcal{W}_{xy}\backslash\{W'\}}W_{xy}\right)W_{n+1},
  \]
  where products are concatenation of walks, and the $n$ occurences of $e$, but also the walks under the products, are oriented in the right ways for the final walk to be well defined (e.g.~$e^1=yx$, $e^2=yxy$ and not $e^2=yxyx$, etc).
  
  \item Assume now that $\mathcal{W}_{xy}$ is empty. If $n$ is odd, then replace $W$ by the equivalent walk
  \[
    W_1\left( \prod_{W_x\in\mathcal{W}_x}W_x \right)e^n\left( \prod_{W_y\in\mathcal{W}_y}W_y \right)W_{n+1}.
  \]
  If $n$ is even, then replace $W$ by the equivalent walk
  \[
    W_1\left( \prod_{W_x\in\mathcal{W}_x}W_x \right)e\left( \prod_{W_y\in\mathcal{W}_y}W_y \right)e^{n-1}W_{n+1}.
  \]
  \end{itemize}
  
  The fact that the parts $W_i$ of $W$ are not modified by this process allows us to apply it recursively to each edge of $W$.
  At the end, we are left with a walk in which each edge appears in at most two places.
\end{proof}

 As we follow an irregularising walk $W$, when seeing a vertex $u$ for the last time, $W$ must get rid of conflicts between $u$ and its neighbours that will not be seen again along $W$, and by Proposition~\ref{proposition:walk_shapes} this must be done through half-turns.
 From this, designing an optimal irregularising walk as in the proof of Theorem~\ref{theorem:bound-general} can be done following a greedy approach in which as little half-turns as necessary are added each step. This approach, provided it is performed along the right ``guiding walk'', is sometimes optimal; this will, notably, be illustrated later on by the cases of paths and cycles (Theorem~\ref{theorem:paths} and Corollary~\ref{corollary:cycles}). This approach is not optimal in general, however: it may happen that adding extra half-turns (\textit{i.e.}, more than the minimum sufficing in the greedy procedure) at some point prevents
 the addition of many of them later on. A first illustration of that is the refinement of Theorem~\ref{theorem:bound-general} and Corollary~\ref{corollary:bound-general} for graphs of bounded chromatic number, which we now provide.

\subsection{Graphs of bounded chromatic number}

In the following, we show that the proof of Theorem~\ref{theorem:bound-general} can be refined for graphs with given chromatic number. Recall that, for a graph $G$, a \textit{proper $k$-vertex-colouring} is a partition of $V(G)$ into $k$ stable sets, and that the \textit{chromatic number} of $G$, denoted $\chi(G)$, is the smallest $k \geq 1$ such that $G$ admits proper $k$-vertex-colourings.

\begin{theorem}\label{theorem:bound-chromatic}
Let $G$ be a nice connected graph of order $n$, maximum degree $\Delta$, and chromatic number $k$,
and $W$ be a closed walk of length~$p$ of $G$ going through all vertices. 
Then, $$\mlw(G) \leq p + (n-1)(2k-2)+2\Delta.$$
\end{theorem}

\begin{proof}
The proof is essentially similar to that of Theorem~\ref{theorem:bound-general}, with a little tweak. Namely, we now consider a proper $\{0,\dots,k-1\}$-vertex-colouring
$\phi$ of $G$, and the half-turns we perform to build $W'$ following $W$ will be w.r.t.~$\phi$.
That is, we partition the set of positive integers into sets $\mathbb{N}_0,\dots,\mathbb{N}_{k-1}$, as follows: we add $1$ and $2$ to $\mathbb{N}_0$, then $3$ and $4$ to $\mathbb{N}_1$, then $5$ and $6$ to $\mathbb{N}_2$, and so on, until we add $2k-1$ and $2k$ to $\mathbb{N}_{k-1}$; we then repeat this process by adding $2k+1$ and $2k+2$ to $\mathbb{N}_0$, $2k+3$ and $2k+4$ to $\mathbb{N}_1$, and so on and so forth to get a partition of all positive integers into $k$ sets (with the property that if $2x \in \mathbb{N}_i$ then $2x-1 \in \mathbb{N}_i$). That is, $x\in\mathbb N_i$ if, and only if, $\lfloor (x+1)/2 \rfloor = (i+1)\bmod k$.
And now, for the walk $W'$ we desire to design from $W$, we require that, for most vertices $u$ of $G$, we have $d_{G+W'}(u) \in\mathbb N_{\phi(u)}$. Obviously, if this was achieved for all vertices, then $W'$ would be irregularising, since for all edges $uv$ of $G$ we would get $d_{G+W'}(u) \in \mathbb N_{\phi(u)}$ and $d_{G+W'}(v) \in \mathbb N_{\phi(v)}$, but $\phi(u) \neq \phi(v)$ since $\phi$ is a proper vertex-colouring of $G$, so $\mathbb N_{\phi(u)} \cap \mathbb N_{\phi(v)} = \emptyset$ and thus $d_{G+W'}(u) \neq d_{G+W'}(v)$. Again, this extra condition can be achieved easily (by adding half-turns to $W'$ while following $W$) for all vertices of $G$ but $u_0=u_p$ and $v_2$; we deal with this issue in the exact same way as in the proof of Theorem~\ref{theorem:bound-general} (through extra modifications we perform along $u_{p-1}u_p$ and $u_pv_2$). In particular, recall how we chose $u_0$ and $v_2$.

So, let us reconsider the same cases as in the proof of Theorem~\ref{theorem:bound-general}. 

\begin{itemize}
	\item If $i < p-1$ and $u_i$ is to be reconsidered later in the process, then, again, we just skip $u_i$ and proceed with $u_{i+1}$. This increases $\|W'\|$ by $1$.
	
	\item If $i < p-1$ and $u_i$ is not going to be reconsidered later in the process, then we first go to $u_{i+1}$. Then, 
	we perform sufficiently many half-turns from $u_{i+1}$ to $u_i$ so that, in $G+W'$, we have $d(u_i) \in \mathbb{N}_{\phi(u_i)}$. Given how $\mathbb{N}_0,\dots,\mathbb{N}_{k-1}$ were constructed, at most $k-1$ such half-turns suffice (indeed, note that the $\mathbb{N}_i$'s contain pairs of consecutive integers, while performing a half-turn increases two degrees by~$2$). So, $\|W'\|$, in total, is increased by at most $2(k-1)+1=2k-1$.
	
	\item For the remaining cases ($i=p-1$ and $i=p$), we can mostly proceed as in the proof of Theorem~\ref{theorem:bound-general}. To begin with, for $i=p-1$ we can first go to $u_p$ and perform half-turns from $u_p$ to $u_{p-1}$ so that $d(u_{p-1}) \in \mathbb{N}_{\phi(u_{p-1})}$ and the resulting degree of $u_p$ is not the same as that of $v_2$. In the worst-case scenario, this requires to perform up to $2k-1$ such half-turns. After that, for $i=p$, it suffices to prolong $W'$ by a walk going back and forth between $u_p$ and $v_2$ (note that any such walk alters their degrees the same way, which thus remain different) until $u_p$ and $v_2$ are not in conflict with their other neighbours in $G+W'$. Note that there is such a walk between $u_p$ and $v_2$ of length at most $d_G(u_p)-1+d_G(v_2)-1 \leq 2\Delta-2$.
\end{itemize}

All in all, we thus have 
\begin{align*}
\|W'\| &\leq \|W\| + (n-3)(2k-2)+2(2k-1)+d(u_p)-1+d(v_2)-1\\
&= p + (n-1)(2k-2)+2\Delta.
\end{align*}

This concludes the proof.
\end{proof}

Again, as for Corollary~\ref{corollary:bound-general}:

\begin{corollary}\label{corollary:bound-chromatic}
If $G$ is a nice connected graph of order $n$, maximum degree $\Delta$, and chromatic number $k$, 
then $$\mlw(G) \leq 2k(n-1)+2\Delta.$$
\end{corollary}

As discussed after Corollary~\ref{corollary:bound-general}, we do not expect Corollary~\ref{corollary:bound-chromatic} to be best possible in general, as there are several ways to optimize our arguments. To see how Corollary~\ref{corollary:bound-chromatic} differs from previous Corollary~\ref{corollary:bound-general}, one should consider dense graphs
having bounded chromatic number. As an example, bipartite graphs ($k=2$) with the most edges are balanced bipartite graphs $K_{n,n}$, which have order $N=2n$ and size $\left(\frac{N}{2}\right)^2=\frac{N^2}{4}$; thus, Corollary~\ref{corollary:bound-general} gives $\mlw(K_{n,n}) \leq N^2$. Meanwhile, we have $\chi(K_{n,n})=2$ and $\Delta(K_{n,n})=n=\frac{N}{2}$; thus, Corollary~\ref{corollary:bound-chromatic} gives $\mlw(K_{n,n}) \leq 5N-4$. For sparser graphs, however, the bounds of Corollaries~\ref{corollary:bound-general} and~\ref{corollary:bound-chromatic} are likely more comparable. 




\section{Particular classes of graphs}\label{section:particular-classes}

We investigate here how far from optimal the bounds provided by Corollaries~\ref{corollary:bound-general} and~\ref{corollary:bound-chromatic} are. As mentioned right after the former corollary,
we are aware that, already, through easy modifications of our proofs, we could establish better bounds. To get a better grasp on how far these two bounds actually are from optimal, we focus in what follows on simple and common classes of graphs.

\subsection{Complete graphs}

For any $n \geq 1$, we denote by $K_n$ the complete graph of order $n$. Recall that, in our context, we restrict ourselves to $n \geq 3$ since we are interested solely in nice graphs. Observe also that the case of complete graphs is interesting in this field, since they are regular. Here, we are able to determine the exact value of $\mlw$.

\begin{theorem}\label{theorem:complete}
We have:
\begin{itemize}
    \item $\mlw(K_3)=3$;
    \item $\mlw(K_n)=\frac{n^2-5n+10}{2}$ for any $n \geq 4$.
\end{itemize}
\end{theorem}

\begin{proof}
Note that an irregularising walk $W$ of minimum length of any graph $G$ must minimise the total number of degree alterations over all vertices; that is, 
the sum of all degrees in $G+W$ must be as small as possible.
Thus, for any irregularising walk $W$ of some complete graph $K_n$, note that, in $K_n+W$, 1) there are, compared to $K_n$, either none (if $W$ is closed) or exactly two (otherwise) vertices having their degrees altered by an odd amount, and 2) any two vertices must have their degrees altered by distinct amounts.

For $K_3$, note that any irregularising walk must go through all vertices. From this and the arguments above, the optimal set of degree alterations for $K_3$ is $\{+1,+2,+3\}$, which is realisable: denoting by $v_1,v_2,v_3$ the vertices, consider $W=v_1v_2v_3v_2$. So, $\mlw(K_3)=3$.

For $n \geq 4$, we proceed by induction on $n$.
We prove that, in any $K_n$, there is an irregularising walk which realises the set of degree alterations $$\{+0,+1,+2,+3,+4,+6,+8,\dots,+2(n-3)\},$$ which is the best we can hope for, by earlier arguments. As a base case, for $K_4$, denoting by $v_1,v_2,v_3,v_4$ the vertices, this set is achieved \textit{e.g.}~by the walk $v_2v_3v_4v_3$. Assume now that the claim holds up to some $n-1$, and consider proving it for $n$. Let us denote by $v_1,\dots,v_n$ the vertices of $K_n$. By the induction hypothesis, in the $K_{n-1}$ subgraph of $K_n$ with vertices $v_1,\dots,v_{n-1}$, there is an irregularising walk $W'$ which realises the set of degree alterations $$\{+0,+1,+2,+3,+4,+6,\dots,+2(n-4)\}.$$ W.l.o.g., assume $W'$ modifies the degree of $v_1,\dots,v_{n-1}$ by $+0,+1,+2,+3,+4,+6,\dots,+2(n-4)$, respectively. 
In particular, the end-vertices of $W'$ are $v_2$ and $v_4$. In $K_n$, we extend $W'$ to a walk $W$ by adding (at the end) the edges of the walk $v_4v_2v_5v_6v_7 \dots v_n$. Then $W$ realises the set of degree alterations $$\{+0,+1,+2,+3,+4,+6,\dots,+2(n-3)\};$$ in particular, compared to $K_n$, the degree of $v_n$ is augmented by $1$, the degree of $v_2$ is augmented by $3$, the degree of $v_4$ is augmented by $4$, the degrees of $v_1$ and $v_3$ remain augmented by $0$ and $2$, respectively, while, for all vertices in $\{v_5,\dots,v_{n-1}\}$, the degrees are augmented by $2$ more compared to how they are augmented by $W'$. Thus, the induction hypothesis holds.

In these inductive arguments, note that the length of $W$ is $n-3$ more than that of $W'$. Since $\mlw(K_4)=3$, the result follows.
\end{proof}

In our opinion, an interesting fact in the proof of Theorem~\ref{theorem:complete} lies in the type of arguments we employ, which consist in reasoning in terms of smallest degrees we can achieve in $G+W$, by an irregularising walk $W$ of some graph $G$. Unfortunately, while such arguments are perhaps usable in some contexts, we doubt they generalise to a large extent. Indeed, complete graphs have the property that all their vertices are pairwise adjacent, and thus that all degrees must be pairwise distinct in $G+W$. Complete graphs are also regular, which makes it easier to comprehend how degrees are affected in $G+W$. In graphs with a wider variety of degrees, it is likely that this is harder to understand, in general.

For comparison, note that for $n \geq 3$, by Corollaries~\ref{corollary:bound-general} and~\ref{corollary:bound-chromatic}, the upper bounds on $\mlw(K_n)$ we deduce are $n^2+n-2$ and $2n^2-2$, respectively. 
Thus, there is some difference between optimal Theorem~\ref{theorem:complete} and these general bounds,
although the proof arguments are actually common.
As mentioned earlier, this is rather expected due to the very restricted structure of complete graphs. 




\subsection{Complete bipartite graphs}

For any $n,m \geq 1$, we denote by $K_{n,m}$ the complete bipartite graph with partition classes of size $n$ and $m$.
Again, in this restricted context, we determine the exact value of $\mlw$.

\begin{theorem}\label{theorem:complete-bipartite}
For any $n,m \geq 1$ with $n+m \geq 3$, we have:
\begin{itemize}
	\item $\mlw(K_{n,m})=0$ if $n \neq m$;
	\item $\mlw(K_{n,m})=2n-2$ if $n=m \geq 2$.
\end{itemize}
\end{theorem}

\begin{proof}
Before starting off, note that the condition $n+m \geq 3$ implies we are dealing with complete bipartite graphs different from $K_2$, thus with nice graphs only.

First off, $K_{n,m}$ is locally irregular whenever $n \neq m$, so $\mlw(K_{n,m})=0$ in those cases.
So, let us now focus on cases where $n=m$.
Recall that $n=m \geq 2$, so $K_{n,m}$ is always nice.

Set $G=K_{n,n}$ for some $n \geq 2$.
We denote by $(U,V)$ the bipartition of $G$, and set $U=\{u_1,\dots,u_n\}$ and $V=\{v_1,\dots,v_n\}$.
Consider any irregularising walk $W$ of $G$.
If there are two vertices $u \in U$ and $v \in V$ which are not traversed by $W$, then, in $G+W$, we have $d(u)=d(v)$, a contradiction.
Since $G$ is a balanced bipartite graph, this implies $\|W\| \geq 2n-2$.
We now give an irregularising walk of length $2n-2$. Start from $v_1$, go to $u_1$, then to $v_2$, then back to $u_1$, then to $v_3$, then back to $u_1$, and so on until reaching $v_n$. This walk $W$ has length $2n-2$. For $n=2$, in $G+W$, all $v_i$'s have odd degree, and all $u_i$'s have even degree, so $W$ is irregularising.
For $n \geq 3$, in $G+W$, all $v_i$'s have degree either $n+1$ or $n+2$, and all $u_i$'s have degree either $n$ or $2n > n+2$, so $W$ is irregularising.
\end{proof}

The proof of Theorem~\ref{theorem:complete-bipartite} is interesting in our opinion, because, just as in that of Theorem~\ref{theorem:complete}, we wonder about the minimum length of an irregularising walk. As mentioned earlier, however, the arguments we employ mostly make sense because of the graph structure considered, and we are not sure they generalise well.

Here, for comparison, note that, for $n \geq 2$, the upper bounds on $\mlw(K_{n,n})$ we derive from Corollaries~\ref{corollary:bound-general} and~\ref{corollary:bound-chromatic} are $2n^2+2n-2$ and $10n-4$, respectively. This illustrates that the bound from Corollary~\ref{corollary:bound-chromatic} is sometimes much better than that from Corollary~\ref{corollary:bound-general}.




\subsection{Paths and cycles}

In what follows, for any $n \geq 1$, for convenience we denote by $P_n$ the path of length~$n$ (\textit{i.e.}, with $n$ edges and, thus, $n+1$ vertices),
and, for any $n \geq 3$, we denote by $C_n$ the cycle of length~$n$. Below, we determine the value of $\mlw$ for paths and cycles. We start by refining Proposition~\ref{proposition:walk_shapes} in the case of paths.

\begin{figure}[!t]
\begin{center}
    \centering
    \includegraphics[scale=0.8]{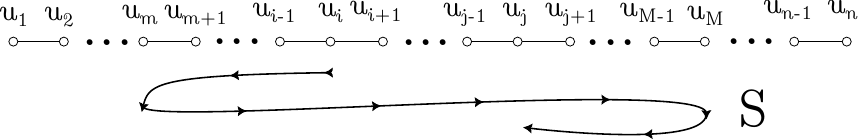}     
    
    \caption{The kind of walk desired in Lemma~\ref{lemma:shape_walk_path}.
    \label{figure:shapeS}}
\end{center}
\end{figure}

\begin{lemma}\label{lemma:shape_walk_path}
Let $n \geq 2$, and let $W$ be a walk of $P_n=u_0 \dots u_n$. Then, there exists a walk of $P_n$ equivalent to $W$ which follows a walk $S$ of the form $$ S=u_iu_{i-1}\dots u_mu_{m+1}\dots u_Mu_{M-1}\dots u_{j}$$ with $m \leq i\le j \leq M$ (see Figure~\ref{figure:shapeS}).
\end{lemma}

\begin{proof}
We denote by $m$ and $M$ the minimum and maximum, respectively, index of a vertex of $W$, and by $i$ and $j$ the index of the first and last vertices of $W$, respectively. Up to traversing $W$ in reverse order, we can assume $i \leq j$.
Obviously, $W$ only contains edges of $P_n$ that are comprised between $u_m$ and $u_M$, and the subpath $u_m\dots u_M$ can be divided into the three pairwise edge-disjoint parts $u_m\dots u_i$, $u_i\dots u_j$, and $u_j\dots u_M$.
 
For any edge $u_{k}u_{k+1}$ of $P_n$, we denote by $t_k \geq 0$ the number of times this edge is traversed by $W$. Note that, by definition of $m$ and $M$, we have $t_k \geq 1$ when $m \leq k \leq M-1$, and $t_k=0$ otherwise. By induction on the length of $W$, one can check that:

\begin{itemize}
    \item $t_k$ is even for all $m\le k\le i-1$;
    \item $t_k$ is odd for all $i\le k\le j-1$;
    \item $t_k$ is even for all $j\le k\le M-1$.
\end{itemize}

\noindent These parity conditions, combined with the fact that $t_k$ is non-zero exactly when $m\le k\le M-1$, give that the following walk
is well defined and equivalent to $W$:
\begin{align*}      
    u_iu_{i-1}\dots u_m(u_{m+1}u_m)^{t_m/2-1}u_{m+1}\dots u_{i-1}(u_iu_{i-1})^{t_{i-1}/2-1}u_i\\
    (u_{i+1}u_i)^{(t_i-1)/2}u_{i+1}\dots u_{j-1}(u_ju_{j-1})^{(t_{j-1}-1)/2}u_j\\
    u_{j+1}\dots u_M(u_{M-1}u_M)^{t_{M-1}/2-1}u_{M-1}\dots u_{j+1}(u_ju_{j+1})^{t_j/2-1}u_j.
\end{align*}
It is also straightforward that this walk follows $S=u_iu_{i-1}\dots u_mu_{m+1}\dots u_Mu_{M-1}\dots u_{j}$.
\end{proof}

In our main result below, we will use the fact that the problem of making graphs locally irregular via edge multiplications (with no walk requirement) is easy for paths, even if we additionally require that the least possible number of edges are added. This can easily be retrieved from previous works on proper labellings, such as~\cite{BFN22,CLWY11}. We give a sketch of proof providing only details that will be needed later on. For a graph $G$ and a multiset $F$ of edges of $G$, we say that $F$ is \textit{irregularising} if $G+F$ is locally irregular, and we denote by $\varphi(G)$ the smallest cardinality of an irregularising multiset of edges of $G$.

\begin{lemma}\label{lemma:varphi_paths}
For any $n \geq 2$, we have
  \[
    \varphi(P_n)=
    \begin{cases}
        \frac{n}{2}\text{ if }n \equiv 0\bmod 4;\\
        \frac{n}{2}-1\text{ if }n \equiv 2\bmod 4;\\
        \frac{n-1}{2}\text{ otherwise}.
    \end{cases}
  \]
\end{lemma}

\begin{proof}
By arguments from \textit{e.g.}~\cite{BFN22,CLWY11}, an optimal irregularising multiset of $P_n=u_0 \dots u_{n}$
is the set containing $u_{4i+2}u_{4i+3}$ and $u_{4i+3}u_{4i+4}$ for all $i \geq 0$. This yields the stated equalities. 
\end{proof}

We are now ready for our main result in this section.

\begin{theorem}\label{theorem:paths}
For any $n \geq 2$, we have
  \[
    \mlw(P_n)=
    \begin{cases}
      0\text{ if }n=2;\\
      1\text{ if }n=3;\\
      2\text{ if }n \in \{4,5\};\\
      2n-10\text{ otherwise.}\\
    \end{cases}
  \]
\end{theorem}

\begin{proof}
Set $P=P_n=u_0 \dots u_{n}$.
Since the bounds can easily be checked when $n \leq 9$, from now on we focus on cases with $n \geq 10$.
Let us consider an optimal irregularising walk $W_S$ of $P$, which can be assumed to meet the form described in Lemma~\ref{lemma:shape_walk_path}. Reusing the terminology from that lemma, $W_S$ divides $P$ into exactly five different part: parts $u_0\dots u_m$ and $u_M\dots u_{n}$ have no edges traversed by $W_S$; parts $u_m\dots u_i$ and $u_j\dots u_M$ have all their edges traversed a non-zero, even number of times by $W_S$;
and part $u_i\dots u_j$ has all its edges traversed a non-zero, odd number of times by $W_S$. Given a part $Q$ of $P$,
we will denote by $k_Q$ the minimum number of times each edge of this part has to be traversed by $W_S$ ($0$, $2$, and $1$, respectively, for the three types of parts we described above). This provides parity conditions
on the number of added parallel edges in the five different parts.

In particular, all edges but those of $u_i\dots u_j$ are traversed an even number of times by $W_S$. As a consequence, we get the following useful fact we will use implicitly throughout:

\begin{remark}
Whenever $i \neq 0$, $i\neq j$, and $j\neq n$, then, in $P+W_S$, all vertices have even degree, except exactly for
$u_0$, $u_i$, $u_j$, and $u_n$.
\end{remark}

The parity conditions also ensure that,
given any edge $e$ of any part $Q$ of $P$, we can define $p(e)$ to be the nonnegative integer such that $2p(e)+k_Q$ is the number of edges $W_S$ adds on top of $e$. If $e$ is part of $Q$, then we set $k(e)=k_Q$.
More formally, Lemma~\ref{lemma:shape_walk_path} states that
\begin{align*}
    W_S=&u_i\dots u_m\\
    &u_{m+1}(u_mu_{m+1})^{p(u_mu_{m+1})}\dots u_{M}(u_{M-1}u_M)^{p(u_{M-1}u_M)}\\
    &u_{M-1}\dots u_j,
\end{align*} 
meaning that $W_S$ performs $p(u_ku_{k+1})$ half-turns on $u_ku_{k+1}$ on top of following
\[
    S=u_iu_{i-1}\dots u_mu_{m+1}\dots u_Mu_{M-1}\dots u_j.
\]

Changing any value $p(u_ku_{k+1})$ gives a new walk.
We do such local modifications to $W_S$ to prove that $m$, $i$, $j$, and $M$ must have restricted values.
First, we prove $m=2$ and $M=n-2$. If we had $m\ge 3$, then $u_1$ and $u_2$ would both have degree $2$ in $P+W_S$,
a contradiction. Hence, $m\le 2$. So assume $m<2$ and consider the walk 
\begin{align*}
    W_S' = & u_{\max \{i,2\}}\dots u_2\\
    &u_{3}(u_2u_3)^{p(u_2u_{3})}\dots u_{M}(u_{M-1}u_M)^{p(u_{M-1}u_M)}\\
    &u_{M-1}\dots u_j.
\end{align*}
Consider any two neighbours $u_k$ and $u_{k+1}$ in $P$. If $k\ge 3$, then the degrees of $u_k$ and $u_{k+1}$ in $P+W_S'$
are the same as in $P+W_S$, and hence they are distinct. For the remaining cases,
\begin{align*}
    d_{P+W_S'}(u_0)=& ~1,\\
    d_{P+W_S'}(u_1)=& ~2,\\
    d_{P+W_S'}(u_2)=& ~2+k(u_2u_3)+2p(u_2u_3),\\
    d_{P+W_S'}(u_3)=& ~2+k(u_2u_3)+2p(u_2u_3)+k(u_3u_4)+2p(u_3u_4).
\end{align*}
As $n\ge 10$, we necessarily have $k(u_3u_4)>0$, which proves that $W_S'$ is irregularising. This is a contradiction as $W_S'$ is shorter
than $W_S$. Hence, $m=2$ and by symmetry $M=n-2$.


We now prove that $i\le 4$ and $j\ge n-4$. Towards a contradiction, assume that $i\ge 5$. In particular $i-3\ge m$. We claim that the number of parallel edges added
by $W_S$ on top of $u_{i-3}u_{i-2}u_{i-1}u_i$ is at least $8$.
Indeed, since $i-3\ge m$, the parity conditions ensure that, for any edge $e$ of this subpath, we have $k(e)=2$.
Walk $W_S$ must perform at least one half-turn on top of one of these edges, as otherwise $u_{i-2}$ and $u_{i-1}$ would
both have degree $6$ in $P+W_S$. Hence, in total $W_S$ adds at least $8$ edges
on top of $u_{i-3}u_{i-2}u_{i-1}u_i$.

Towards a contradiction, we prove that we can build an irregularising walk $W_S'$ of shorter length
by case analysis on the values of $p(u_iu_{i+1})$ and  $p(u_{i+1}u_{i+2})$. In each
case, $W_S'$ is constructed by modifying values $p(e)$ for edges $e$ of $u_{i-3}u_{i-2}u_{i-1}u_i$
and by making this subpath part of the central part of $P$. We consider some cases:
\begin{itemize}
    \item In case $p(u_i u_{i+1})=0$ and $p(u_{i+1}u_{i+2})=0$, consider (see Figure~\ref{figure:paths_cases})
    \begin{align*}
    W_S'=~&u_{i-3}\dots u_m\\
    &u_{m+1}(u_mu_m+1)^{2p(u_mu_{m+1})}\dots u_{i-3}(u_{i-4}u_{i-3})^{2p(u_{i-4}u_{i-3})}\\
    &u_{i-3}u_{i-2}u_{i-1}(u_{i-2}u_{i-1})u_i(u_{i-1}u_i)\\
    &u_{i+1}(u_i u_{i+1})^{p(u_i u_{i+1})}\dots u_{M}(u_{M-1}u_{M})^{p(u_{M-1}u_M)}\\
    &u_{M-1}\dots u_j.
    \end{align*}
    In $P+W_S'$, consider any two neighbours $u_k$ and $u_{k+1}$. If $k+1\le i-4$ or $k\ge i+1$, then, since the degrees in $P+W_S'$ of $u_k$ and $u_{k+1}$ are the same as in $P+W_S$,
    and because $W_S$ is irregularising, we have $d_{P+W_S'}(u_k)\neq d_{P+W_S'}(u_{k+1})$. Also, $d_{P+W_S'}(u_{i-2})=6$, $d_{P+W_S'}(u_{i-1})=8$, $d_{P+W_S'}(u_i)=6$, and $d_{P+W_S'}(u_{i+1})=4$ (since $p(u_i u_{i+1})=0$ and $p(u_{i+1}u_{i+2})=0$).
    Finally, when $k=i-3$ and $k=i-4$, the degree of $u_{i-3}$ in $P+W_S'$ is odd
    while that of $u_{i-2}$ and $u_{i-4}$ are even. Hence, $W_S'$ is irregularising, and, by definition, its length is less than that of $W_S$ (since we removed at least~$8$ edges, and added exactly~$7$), a contradiction.
    
    \item The other cases are illustrated in Figure~\ref{figure:paths_cases} and are treated in similar ways.
\end{itemize}

\begin{figure}
    \centering
    \includegraphics{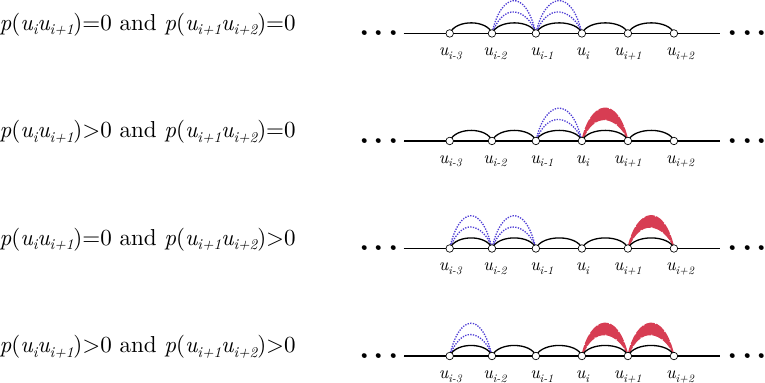}
    \caption{Cases considered in the proof of Theorem~\ref{theorem:paths}. There are four cases depending on whether $W_S$ performs a non-zero number of half-turns on top of $u_iu_{i+1}$ and $u_{i+1}u_{i+2}$. Red filled crescents represent these half-turns. Straight black regular lines are the edges of $P$ while curved black regular lines represent edges added by $W_S'$ besides half-turns (here, we display the central part $Q$ of $W_S'$ so that $k_Q=1$). Blue double edges are half-turns performed on top of $u_{i-3}u_{i-2}u_{i-1}u_i$ by $W_S'$. Hence, these figures illustrate the exact behaviour of $W_S'$ on the part $u_{i-3}u_{i-2}u_{i-1}u_i$ of $P$ along which $W_S'$ differs from $W_S$.}
    \label{figure:paths_cases}
\end{figure}

So, we can now assume $i\le 4$, and $j\ge n-4$ by symmetry. As $n\ge 10$, this implies $j>i+1$. We prove that $p(u_ku_{k+1})=0$ for every $m\le k\le i-1$.
Let $W_S'$ be the walk obtained from $W_S$ by setting these half-turn numbers to be
$0$ while
retaining the other performed half-turns. Consider any two neighbours $u_k$
and $u_{k+1}$ of $P$. If $k\ge i+1$ or $k=0$, then the degrees in $P+W_S'$
of $u_k$
and $u_{k+1}$ are the same as in $P+W_S$.
If $k=1$, then we have $d_{P+W_S'}(u_1)=2$ and $d_{P+W_S'}(u_2)\ge 3$ as
$W_S'$ must traverse $u_2u_3$. In case $k=2$, the degree
of $u_3$ is the sum of that of $u_2$ plus the number of parallel
edges that $W_S'$ adds on top of $u_3u_4$ (which is non-zero as $n\ge 10$).
Finally, if $k=i-1$ or $k=i$, then the parity conditions when $j>i+1$ ensure that $u_i$
has odd degree while its neighbours have even degrees.

As $j>i+1$, the central part $Q=u_iu_{i+1}\dots u_j$ of $P$
is nice (\textit{i.e.}, is not $K_2$). As $W_S$ is irregularising, we get that the multiset $\{(u_ku_{k+1},p(u_ku_{k+1}))\mid i\le k\le j-1\}$
is irregularising for $Q$. Combining Lemma~\ref{lemma:varphi_paths} and facts proved earlier in the proof provides the following bound:
\begin{align*}
    \mlw(P)\ge&~ 2(i-m)+(j-i)+2\varphi(P_{j-i})+2(M-j)\\
    \ge&~2n+i-j+2\varphi(P_{j-i})-8.
\end{align*}

We prove that equality actually holds. Consider an irregularising multiset $F$ for $u_i\dots u_j$, and denote by $p_F(u_ku_{k+1})$
the number of edges $u_ku_{k+1}$ it contains for $i\le k\le j-1$. Then, the walk
\begin{align*}
    W_S'=~&u_i\dots u_m\dots u_i\\
    &u_{i+1}(u_iu_{i+1})^{p_F(u_iu_{i+1})}\dots u_{j}(u_{j-1}u_j)^{p_F(u_{j-1}u_j)}\\
    &u_{j+1}\dots u_M\dots u_j
\end{align*}
is irregularising for $P$. Indeed, the degrees of any two neighbours $u_k$ and $u_{k+1}$
in $P+W_S'$ are distinct whenever $k\in\{0,\dots,i-2\}\cup\{j+1,\dots,n-1\}$ as their
degrees are the same as in $P+W_S$, and $W_S$ is irregularising.
Since $j>i+1$, the degrees of $u_i$ and $u_j$ are odd while those of their neighbours are even,
so whenever $k\in\{i-1,i,j-1,j\}$ the two neighbours have distinct degrees. The other cases follow directly
from $F$ being irregularising. Then, the fact that $W_S'$ is irregularising for $P$ proves the other
inequality; so, for some value of $i$ and $j$:
\begin{align*}
    \mlw(P)=&~2n+i-j+2\varphi(P_{j-i})-8.
\end{align*}
Now, since $2\le i\le 4$, $n-4\le j\le n-2$, and, up to symmetry between $i$ and $j$, there are $6$ possible pairs of values to consider, by
exploring the possible values of $n$ modulo $4$ we get the final statement by Lemma~\ref{lemma:varphi_paths}.
\end{proof}

We believe an interesting thing behind the proof of Theorem~\ref{theorem:paths} is that, through using Lemmas~\ref{lemma:shape_walk_path} and~\ref{lemma:varphi_paths}, we build upon results on proper labellings. Although this cannot be applied for all graphs, recall our discussion in Section~\ref{subsection:relation-labellings}, this proves that, still, there are situations where this philosophy applies.

Regarding Theorem~\ref{theorem:paths}, for comparison, for any $n \geq 2$, the bounds on $\mlw(P_n)$ we derive from Corollaries~\ref{corollary:bound-general} and~\ref{corollary:bound-chromatic} are $4n$ and $4n+4$, respectively. 

A corollary we deduce is:

\begin{corollary}\label{corollary:cycles}
For any $n \geq 3$, we have
  \[
    \mlw(C_n)=
    \begin{cases}
      3\text{ if }n=3;\\
      2n-6\text{ otherwise.}
    \end{cases}
  \]
\end{corollary}

\begin{proof}
The result is trivial for $n=3$. We assume that $n\ge 4$ from now on. Let $W$ be any irregularising walk on $C_n$. Let $G_e$, $P_o$, and $P_\neg$ be the subgraphs
of $C_n$ induced by those edges that are traversed, respectively, an even non-zero number of times by $W$, an odd number of times by $W$, and that are not traversed by $W$. Note that some of $G_e$, $P_o$, and $P_\neg$ might be empty, and that the next arguments still apply even when this is the case. Obviously, $C_n$ is the edge-disjoint
union of $G_e$, $P_o$, and $P_\neg$.
By induction on the length of $W$,
one can see that the first and last vertices of $W$ must be the end-vertices of $P_o$, that $P_o$ and $P_\neg$ are connected,
and that $G_e$ is made of at most two connected components $P_e$ and $P_e'$ with, for both of them, one end-vertex shared with $P_o$ and the other with $P_\neg$. In particular,
$P_e$, $P_e'$, $P_o$, and $P_\neg$ are all subpaths of $C_n$ and $C_n=P_eP_oP_e'P_\neg$. Similarly as in the proofs of Proposition~\ref{proposition:walk_shapes} and Lemma~\ref{lemma:shape_walk_path},
by simply adding the necessary half-turns, one can then prove that $W$ must follow the walk
\[
    \overline{P_e}P_eP_oP'_e\overline{P'_e},
\]
where $\overline{P}$ is the path $P$ in reverse order, and where ``multiplications of paths'' are just concatenations of paths.

We prove that if $W$ is irregularising and optimal, then $P_\neg$ has length exactly $2$. First off, it cannot
have length at least $3$ as then any two inner neighbours of $P_\neg$ would have equal degree in $P+W$.
Now suppose $P_{\neg}$ has length $1$, and write $P_{\neg} = uv$. Let $w$ be the other neighbour of $u$ in $C_n$. Then, removing from $W$ all occurrences of $uw$ yields a shorter walk $W'$. Note that all vertices have the same degree in $G+W$ and $G+W'$, except for $u$ and $w$. Furthermore, denoting by $x$ the neighbour of $w$ other than $u$, in $G+W'$, we have $d(u) = 2$, $d(v) \geq 3$, $d(w) \geq 3$, and $d(x) \geq d(w)+1$ (using here that $n \geq 4$); so $W'$ is irregularising, a contradiction.
Note that we can do something similar if $P_\neg$ had length $0$ (we voluntarily omit giving a formal proof, which would be tedious as it would have to consider whether $P_e$, $P_e'$, and $P_o$ are empty). Thus we get a contradiction whenever $P_{\neg}$ has length different from~$2$.



Assume now $P_\neg$ has length~$2$. Due to this, its inner vertex acts just as the second and penultimate vertices of a
path of length $n+2$ so that $W$ is irregularising and optimal for $C_n$ if, and only if, it corresponds
to an optimal irregularising walk of $P_{n+2}$ of the same length (see Figure~\ref{figure:cycles}). As $n+2\ge 6$, the result follows
from Theorem~\ref{theorem:paths}.
\begin{figure}
    \centering
    \includegraphics{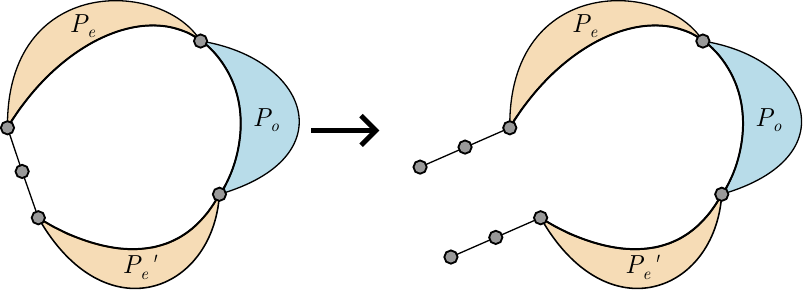}
    \caption{Arguments used in the proof of Corollary~\ref{corollary:cycles}. On the left is a cycle $C$ of length $n$ with an associated walk $W$. Only the common extremities of $P_e, P'_e$ and $P_o$ as well as $P_\neg$ are explicitly represented using $5$ gray vertices. Parallel edges added on top of $P_e, P'_e$ and $P_o$ by $W$ are represented by the three thick coloured crescents.
    On the right is a path $P$ of length $n+2$ built from the cycle. The possible degree conflicts between neighbours in $P$ and in $C$ are exactly the same.}
    \label{figure:cycles}
\end{figure}
\end{proof}




\section{Algorithmic aspects}\label{section:algo}

In this section, we provide algorithmic results on the complexity of determining $\mlw$. We begin with negative results, by showing, in brief, that determining $\mlw(G)$ for a given graph $G$ is \np-complete, even when $G$ belongs to restricted graph classes. In the same vein, we then establish that determining whether a given graph admits irregularising paths is \np-complete. On the positive side of things, we mainly prove that determining $\mlw(T)$ for a given tree $T$ can be done in polynomial time.

\subsection{Determining whether a graph admits irregularising walks}

From the algorithmic point of view, 
a first natural question is whether it is easy to determine if $\mlw(G) \leq k$ holds for some given graph $G$ and fixed $k \geq 1$.
Formally:

\medskip

\noindent\textsc{Irregularising $k$-Walk}\\
\noindent\textbf{Input:} A graph $G$.\\
\noindent\textbf{Question:} Do we have $\mlw(G) \leq k$?

\medskip

Clearly, \textsc{Irregularising $k$-Walk} lies in \textsf{P} for all $k \geq 1$: one can simply enumerate all walks $W$ of length at most $k$
of $G$ (there are $\mathcal{O}(m^k)$ of them) and check whether any of them is irregularising by looking at degrees in $G+W$. 

The next step is thus to wonder what happens if $k$ is part of the input:

\medskip

\noindent\textsc{Irregularising Walk}\\
\noindent\textbf{Input:} A graph $G$, and some $k \geq 1$.\\
\noindent\textbf{Question:} Do we have $\mlw(G) \leq k$?

\medskip

The previous exhaustive search algorithm shows that \textsc{Irregularising Walk} lies in \textsf{XP} when parameterised by $k$, which leads to the question of whether it lies in \textsf{FPT}, which we will consider later on.

We prove that \textsc{Irregularising Walk} is \textsf{NP}-complete.

\begin{theorem}\label{theorem:npc-walks}
\textsc{Irregularising Walk} is \textsf{NP}-complete.
\end{theorem}

\begin{proof}
The problem is in \textsf{NP} since given a sequence $W$ of edges of $G$,
one can check in polynomial time whether $W$ is a walk of length at most~$k$ of $G$,
and whether $G+W$ is locally irregular. Thus, we focus on proving the \textsf{NP}-hardness of \textsc{Irregularising Walk}.
This is done by reduction from the \textsc{Hamiltonian Cycle} problem restricted to cubic bipartite graphs,
which is \textsf{NP}-complete~\cite{ANS80}.
Given a cubic bipartite graph $H$, we construct, in polynomial time,
a graph $G$ such that, for some $k \geq 1$, there is an irregularising walk of length at most $k$ of $G$ (that is, $\mlw(G) \leq k$) if and only if $H$ has a Hamiltonian cycle.

Since $H$ is bipartite, it admits a bipartition $(U,V)$ which we can obtain in polynomial time (\textit{e.g.}~through running a BFS algorithm). 
The construction of $G$, now, goes as follows (a similar construction is illustrated in Figure~\ref{figure:redux-paths}).

\begin{itemize}
	\item Start from $G$ being $H$. 
	
	\item For all vertices $u \in U$, we make, in $G$, vertex $u$ adjacent to two new vertices $a_u$ and $b_u$.
	We then add four leaves adjacent to $a_u$, and five leaves adjacent to $b_u$.
	As a result, we have $d(u)=d(a_u)=5$, and $d(b_u)=6$.
	
	\item For all vertices $v \in V$, likewise we make $v$ adjacent to three new vertices $a_v$, $b_v$, and $c_v$.
	We then add five leaves adjacent to $a_v$, and six leaves adjacent to $b_v$.
	As a result, $d(v)=d(a_v)=6$, $d(b_v)=7$, and $d(c_v)=1$.
\end{itemize}

\noindent The construction of $G$ is achieved in polynomial time. 
The value of $k$ we choose for the instance is $n$, where $n=|V(H)|$.

To see that we have the desired equivalence between $H$ and $G$ (with $k=n$), let us explore how an irregularising walk $W$ of length at most $k=n$ of $G$ should behave.
First off, note that $G$ is not locally irregular.
More precisely, all degree conflicts are along edges of the form $ua_u$ (for all $u \in U$) or $va_v$ (for all $v \in V$).
As a result, for all $u \in U$ it must be that $W$ goes through $u$ or $a_u$, and for all $v \in V$ it must be that $W$ goes through $v$ or $a_v$. Note that, because $|V(H)|\ge2$, for every $u\in V(H)$, if $W$ goes through $a_u$ then it also goes through $u$ to get to other vertices $v\in V(H)$ or $a_v$ (thus, distinct from $u$).
From this, we get that $W$ must go through every vertex of $V(H)$.

So the edges of $H$ traversed by $W$ induce a connected spanning subgraph $K$
of $H$ of size at most $n$ (connected because every connected component of $G[V(G) \setminus V(H)]$ has at most one neighbour in $V(H)$). It is well known that, in this case, either $K$ has size $n-1$ and is a spanning subtree
of $H$, or $K$ has size $n$ and contains exactly one cycle. 

\begin{itemize}
    \item In the latter case, $W$ must only traverse
edges of $K$ once. If $K$ had a leaf $u\in U$ or $v\in V$, then that would imply that $d(u)=d(b_u)$ or $d(v)=d(b_v)$ in $G+W$.
So $K$ must reduce to a cycle going through all vertices of $H$, thus to a Hamiltonian cycle of $H$.
\item In the former case,
$K$ has size $n-1$ and forms a spanning subtree of $H$. Since $H$ is cubic and bipartite we have $|V(H)|\ge 6$. In case
$W$ goes only through vertices of $K$, at most one edge of $K$ can be traversed twice by $W$. Otherwise, all edges of
$K$ are traversed exactly once by $W$, and at most one edge of $G$ not in $E(H)$ can be traversed. Thus, in any case, some leaf $u \in U$ or $v \in V$ of $K$
must be visited only once by $W$, and $d(u)=d(b_u)$ or $d(v)=d(b_v)$ in $G+W$. This is a contradiction.
\end{itemize}

\noindent Hence, the only possibility is that $W$ is a Hamiltonian cycle of $H$.

Last, to see that the other direction of the equivalence holds as well, note that if $W$ is a Hamiltonian cycle of $H$,
then, in $G+W$, vertices of $H$ see their degrees increase by exactly $2$ and vertices not in $H$ retain their original degrees:
in $G+W$,
for all vertices $u\in U$ we have $d(u)=7$, $d(a_u)=5$, and $d(b_u)=6$, and for all vertices in $v\in V$
we have $d(v)=8$, $d(a_v)=6$, and $d(b_v)=7$. All other vertices (including the $c_v$'s) are leaves and
thus have degree $1$ in $G+W$. So $W$ is an irregularising walk of $G$.
\end{proof}

A consequence of the proof of Theorem~\ref{theorem:npc-walks} is that \textsc{Irregularising Walk} remains \textsf{NP}-complete
when restricted to bipartite graphs of maximum degree~$7$. Other such results can be deduced, with some efforts, from the fact that the \textsc{Hamiltonian Cycle} problem remains hard when restricted to particular graph classes. For instance, see~\cite{GJT76}, the \textsc{Hamiltonian Cycle} problem remains \np-hard when restricted to cubic planar graphs; mostly because our graph modifications in the proof of Theorem~\ref{theorem:npc-walks} preserve planarity, our proof can be adapted\footnote{Note that, in such a proof, $H$ would not necessarily be bipartite; however, instead of a bipartition, one can \textit{e.g.}~consider a proper $3$-vertex-colouring $\phi$ of $H$ (which exists by Brooks' Theorem, since $\Delta(H) \leq 3$), and, to build $G$, increase the vertex degrees and add constraining neighbours according to colours by $\phi$. Alternatively, one could just, still in polynomial time, modify all vertex degrees in $H$ so that they are pairwise distinct, and add constraining neighbours accordingly.} to prove that \textsc{Irregularising Walk} remains \textsf{NP}-complete
when restricted to planar graphs of bounded maximum degree.
Likewise, the \textsc{Hamiltonian Cycle} problem is \textsf{W[1]}-hard when parameterised by the clique-width~\cite{FGLS10}, and, again, given the graph modifications we perform in the proof of Theorem~\ref{theorem:npc-walks}, through modifications of our reduction it can be proved that \textsc{Irregularising Walk} is \textsf{W[1]}-hard when parameterised by the clique-width.

The main goal in this paper being to somewhat comprehend the parameter $\mlw$, we naturally focused, above, on the complexity of the decision problems \textsc{Irregularising $k$-Walk} and \textsc{Irregularising Walk}. But one may of course wonder about the complexity of determining $\mew(G)$ and $\mvw(G)$ for an input graph $G$. Given that, despite being slightly different, all these parameters are quite close in spirit, it is possible to tweak the proof of Theorem~\ref{theorem:npc-walks} a bit to establish akin results. Namely, in later Theorem~\ref{theorem:npc-paths}, we will prove (through a slight modification of the proof of Theorem~\ref{theorem:npc-walks}) that determining whether a graph admits an irregularising path is \np-complete. It can be observed that finding an irregularising walk with certain restrictions on how edges/vertices are traversed is actually equivalent to finding an irregularising path. More precisely, through slight modifications of our reduction (similar to ones to be used in the proof of later Theorem~\ref{theorem:npc-paths}), one can prove:

\begin{corollary}\label{cor:np-c}
The following two problems are \np-complete:
\begin{itemize}
    \item determining, for a given graph $G$, whether $\mew(G) \leq 1$;
    \item determining, for a given graph $G$, whether $\mvw(G) \leq 2$.
\end{itemize}
\end{corollary}

Note that this contrasts with the fact that \textsc{Irregularising $k$-Walk} is polynomial-time tractable for all $k \geq 1$.

\subsection{Determining whether a graph admits irregularising paths}

We mentioned earlier that many nice graphs do not admit irregularising paths,
and, even worse, that there is no nice characterisation of them (unless $\p=\np$).
We prove this now, through a slight modification of the reduction in the proof of Theorem~\ref{theorem:npc-walks}.
To be precise, the formal problem we consider reads as:

\medskip

\noindent\textsc{Irregularising Path}\\
\noindent\textbf{Input:} A graph $G$.\\
\noindent\textbf{Question:} Is there an irregularising path of $G$?

\medskip

We thus proceed to proving:

\begin{theorem}\label{theorem:npc-paths}
\textsc{Irregularising Path} is \textsf{NP}-complete.
\end{theorem}

\begin{proof}
The {\np}ness of the problem is obvious.
The \np-hardness part of the result follows from a straight modification of the reduction in the proof of Theorem~\ref{theorem:npc-walks}. The main difference here, is that, since there is here no length restriction, we have to modify the construction a bit to forbid paths traversing all vertices of $H$, but having their two end-vertices in some of the trees we have attached to the $u$'s and $v$'s (this is to forbid paths visiting vertices outside $H$).
Recall that, in the proof of Theorem~\ref{theorem:npc-walks}, this was indeed mainly prevented by the length of the irregularising walks having to be at most $|V(H)|$.

\begin{figure}[t]
\centering

\scalebox{0.8}{
    \begin{tikzpicture}[inner sep=0.7mm,thick]	
    
    \draw[rounded corners,fill=lightgray!50] (-1,1) rectangle (5,-1);
    \node at (2,1.25){$H$};
    
    \node[draw,circle,black,fill=white,line width=1pt](u) at (0,0)[label=above:\scriptsize $u$]{\scriptsize $5$};
    \node[draw,circle,black,fill=white,line width=1pt](au) at (-2,2.5)[label=above:\scriptsize $a_u$]{\scriptsize $5$};
    \node[draw,circle,black,fill=white,line width=1pt](bu) at (-2,-2.5)[label=above:\scriptsize $b_u$]{\scriptsize $6$};
    \node[draw,circle,black,fill=white,line width=1pt](alphau1) at (-5,1)[label=left:\scriptsize $\alpha_u^1$]{\scriptsize $6$};
    \node[draw,circle,black,fill=white,line width=1pt](alphau2) at (-5,2)[label=left:\scriptsize $\alpha_u^2$]{\scriptsize $6$};
    \node[draw,circle,black,fill=white,line width=1pt](alphau3) at (-5,3)[label=left:\scriptsize $\alpha_u^3$]{\scriptsize $7$};
    \node[draw,circle,black,fill=white,line width=1pt](alphau4) at (-5,4)[label=left:\scriptsize $\alpha_u^4$]{\scriptsize $7$};
    \node[draw,circle,black,fill=white,line width=1pt](betau1) at (-5,-1)[label=left:\scriptsize $\beta_u^1$]{\scriptsize $7$};
    \node[draw,circle,black,fill=white,line width=1pt](betau2) at (-5,-2)[label=left:\scriptsize $\beta_u^2$]{\scriptsize $7$};
    \node[draw,circle,black,fill=white,line width=1pt](betau3) at (-5,-3)[label=left:\scriptsize $\beta_u^3$]{\scriptsize $8$};
    \node[draw,circle,black,fill=white,line width=1pt](betau4) at (-5,-4)[label=left:\scriptsize $\beta_u^4$]{\scriptsize $8$};
    
    \node[draw,circle,black,fill=white,line width=1pt](v) at (4,0)[label=above:\scriptsize $v$]{\scriptsize $6$};
    \node[draw,circle,black,fill=white,line width=1pt](av) at (6,2.5)[label=above:\scriptsize $a_v$]{\scriptsize $6$};
    \node[draw,circle,black,fill=white,line width=1pt](bv) at (6,-2.5)[label=above:\scriptsize $b_v$]{\scriptsize $7$};
    \node[draw,circle,black,fill=white,line width=1pt](cv) at (6,0)[label=above:\scriptsize $c_v$]{\scriptsize $6$};
    \node[draw,circle,black,fill=white,line width=1pt](alphav1) at (9,2.5)[label=right:\scriptsize $\alpha_v^1$]{\scriptsize $7$};
    \node[draw,circle,black,fill=white,line width=1pt](alphav2) at (9,3.5)[label=right:\scriptsize $\alpha_v^2$]{\scriptsize $7$};
    \node[draw,circle,black,fill=white,line width=1pt](alphav3) at (9,4.5)[label=right:\scriptsize $\alpha_v^3$]{\scriptsize $8$};
    \node[draw,circle,black,fill=white,line width=1pt](alphav4) at (9,5.5)[label=right:\scriptsize $\alpha_v^4$]{\scriptsize $8$};
    \node[draw,circle,black,fill=white,line width=1pt](betav1) at (9,-2.5)[label=right:\scriptsize $\beta_v^1$]{\scriptsize $8$};
    \node[draw,circle,black,fill=white,line width=1pt](betav2) at (9,-3.5)[label=right:\scriptsize $\beta_v^2$]{\scriptsize $8$};
    \node[draw,circle,black,fill=white,line width=1pt](betav3) at (9,-4.5)[label=right:\scriptsize $\beta_v^3$]{\scriptsize $9$};
    \node[draw,circle,black,fill=white,line width=1pt](betav4) at (9,-5.5)[label=right:\scriptsize $\beta_v^4$]{\scriptsize $9$};
    \node[draw,circle,black,fill=white,line width=1pt](gamma1) at (9,1.5){\scriptsize $7$};
    \node[draw,circle,black,fill=white,line width=1pt](gamma2) at (9,0.5){\scriptsize $7$};
    \node[draw,circle,black,fill=white,line width=1pt](gamma3) at (9,-0.5){\scriptsize $8$};
    \node[draw,circle,black,fill=white,line width=1pt](gamma4) at (9,-1.5){\scriptsize $8$};
    
    \draw [-, line width=2pt,black] (u) -- (au);
    \draw [-, line width=2pt,black] (u) -- (bu);
    \draw [-, line width=2pt,black] (u) -- (1.5,0.75);
    \draw [-, line width=2pt,black] (u) -- (1.5,0);
    \draw [-, line width=2pt,black] (u) -- (1.5,-0.75);
    \draw [-, line width=2pt,black] (au) -- (alphau1);
    \draw [-, line width=2pt,black] (au) -- (alphau2);
    \draw [-, line width=2pt,black] (au) -- (alphau3);
    \draw [-, line width=2pt,black] (au) -- (alphau4);
    \draw [-, line width=2pt,black] (bu) -- (betau1);
    \draw [-, line width=2pt,black] (bu) -- (betau2);
    \draw [-, line width=2pt,black] (bu) -- (betau3);
    \draw [-, line width=2pt,black] (bu) -- (betau4);
    \draw [-, line width=2pt,black] (v) -- (av);
    \draw [-, line width=2pt,black] (v) -- (bv);
    \draw [-, line width=2pt,black] (v) -- (cv);
    \draw [-, line width=2pt,black] (v) -- (2.5,0.75);
    \draw [-, line width=2pt,black] (v) -- (2.5,0);
    \draw [-, line width=2pt,black] (v) -- (2.5,-0.75);
    \draw [-, line width=2pt,black] (av) -- (alphav1);
    \draw [-, line width=2pt,black] (av) -- (alphav2);
    \draw [-, line width=2pt,black] (av) -- (alphav3);
    \draw [-, line width=2pt,black] (av) -- (alphav4);
    \draw [-, line width=2pt,black] (bv) -- (betav1);
    \draw [-, line width=2pt,black] (bv) -- (betav2);
    \draw [-, line width=2pt,black] (bv) -- (betav3);
    \draw [-, line width=2pt,black] (bv) -- (betav4);
    \draw [-, line width=2pt,black] (cv) -- (gamma1);
    \draw [-, line width=2pt,black] (cv) -- (gamma2);
    \draw [-, line width=2pt,black] (cv) -- (gamma3);
    \draw [-, line width=2pt,black] (cv) -- (gamma4);
    \end{tikzpicture}
}
	
\caption{Illustration of the reduction in the proof of Theorem~\ref{theorem:npc-paths}, being a modification of the reduction in the proof of Theorem~\ref{theorem:npc-walks}. For convenience, not all edges are depicted; numbers in vertices indicate their degrees, all missing edges going to leaves.}
\label{figure:redux-paths}
\end{figure}
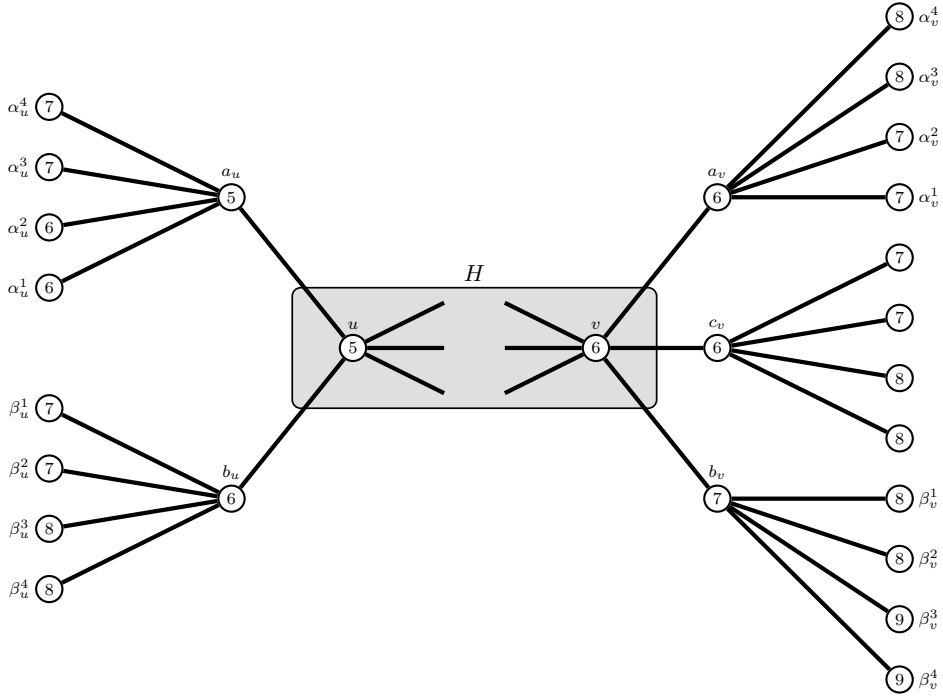

The modified reduction goes as follows (see Figure~\ref{figure:redux-paths}).

\begin{itemize}
    \item Start again from $G$ being $H$.
    
    \item For every vertex $u \in U$, again, make $u$ adjacent to two new vertices, $a_u$ and $b_u$, and make these of degree~$5$ and~$6$, respectively, by making them adjacent to new leaves. Pick now any two leaves $\alpha^1_u$ and $\alpha^2_u$ adjacent to $a_u$, and make them of degree~$6$ by adding new leaves. Likewise, pick any two other leaves $\alpha_u^3$ and $\alpha_u^4$ adjacent to $a_u$, and make them of degree~$7$. We perform similar modifications for neighbours of $b_u$: we make two leaves $\beta^1_u$ and $\beta^2_u$ of degree~$7$, and two others $\beta^3_u$ and $\beta^4_u$ of degree~$8$. 
    
    \item We perform similar modifications around the $v$'s. Namely, as in the proof of Theorem~\ref{theorem:npc-walks}, we make every vertex $v \in V$ adjacent to three new neighbours $a_v$, $b_v$, and $c_v$. Then, we make $a_v$, $b_v$, and $c_v$ of degree~$6$, $7$, and~$6$, respectively. We further modify the construction by making two  leaves $\alpha_v^1$ and $\alpha_v^2$ neighbouring $a_v$ of degree~$7$, and two others $\alpha_v^3$ and $\alpha_v^3$ of degree~$8$. We perform the exact same modifications around $c_v$.
    Eventually, we also turn four leaves $\beta_v^1,\beta_v^2,\beta_v^3,\beta_v^4$ adjacent to $b_v$ into vertices of degree $8$, $8$, $9$, and $9$, respectively.
\end{itemize}

The reduction is still performed in polynomial time.
The desired equivalence, now, follows mainly from the same arguments as in the proof of Theorem~\ref{theorem:npc-walks}. The main difference, here, is that since there is no restriction on the length of a desired irregularising path $P$ of $G$, nothing prevents $P$ from visiting a vertex outside $H$.
However, given the vertices we have added, that $P$ is a path, and since we still have the property that all vertices of $G$ from $H$ must be traversed by $P$, it holds that:

\begin{itemize}
    \item If $P$ starts/finishes at some vertex $a_u$, then it cannot go through any of $\alpha_u^1,\alpha_u^2,\alpha_u^3,\alpha_u^4$; in particular, in $G+P$ we have $d(a_u)=6=d(\alpha_u^1)$, a contradiction. Similar arguments apply regarding the $a_v$'s and $c_v$'s.
    
    \item Likewise, if $P$ starts/finishes at some $b_u$, then, necessarily, in $G+P$ we must have $d(b_u)=7=d(\beta_u^1)$. Again, a similar argument applies regarding the $b_v$'s.
    
    \item We claim that we get a similar contradiction in case $P$ starts/finishes at some vertex that is neither some $u \in U$, $v\in V$, $a_u$, $b_u$, $a_v$, $b_v$, or $c_v$. In that case, note that there must be, say, some $a_u$ such that $P$ goes through $u$, $a_u$, and some neighbour of $a_u$. This means that, in $G+P$, the degree of $a_u$ is altered by~$2$, and is thus equal to~$7$. Given the structure of $G$, it cannot be, however, that $P$ goes through both $\alpha_u^3$ and $\alpha^u_4$; this means that one of these two vertices still has degree~$7$ in $G+P$, and we thus have a conflict with $a_u$, a contradiction. The same goes regarding $b_u$'s, $a_v$'s, $c_v$'s, and $b_v$'s.
\end{itemize}

Thus, as earlier, $P$ must traverse vertices of $H$ only, and, actually, to avoid any conflict with some $b_u$ or $b_v$, it must be that $P$ is a cycle going through all vertices of $H$, thus a Hamiltonian cycle. The desired equivalence now follows immediately.
\end{proof}

\subsection{Determining $\mlw$ in polynomial time for trees}

We here mostly prove the following result:

\begin{theorem}\label{theorem:trees-polynomial}
Determining $\mlw(G)$ can be done in polynomial time when $G$ is a tree.
\end{theorem}

\noindent Proving Theorem~\ref{theorem:trees-polynomial} in details is a bit tedious, which is why we only provide general arguments in what follows.

Let $T$ be a rooted tree. For a given (non-negative integer) weight $w\ge 0$, we generalise the definition of
irregularising walks $W$ to that of \textit{$w$-irregularising walks} by modifying the irregularity constraint between
the root $r$ and its neighbours. Namely, a $w$-irregularising walk $W$ satisfies that, for every neighbour $u$ of $r$, we have $d_{T+W}(r)+w\neq d_{T+W}(u)$.
Hence, $w$ somewhat acts as some ``outside'' edges incident to the root, assuming that $T$ is a subtree rooted at $r$ of a bigger tree. A usual irregularising walk of $T$ is then simply a $0$-irregularising walk.
We then extend the definition of $\mlw(T)$ to $\mlw_{w,d}(T)$ for $w,d\ge 0$ non-negative integers, as the
minimum length of a $w$-irregularising walk $W$ of $T$ such that $d_{T+W}(r)+w=d$ (obviously, we will only treat cases where $d_{T}(r)+w \leq d$).
Hence, $d$ can be thought as the objective full degree of $r$ in a hypothetical bigger tree.

On top of that, we also introduce the following restricted parameters. All are defined
to $+\infty$ when no corresponding irregularising walks exist.

\begin{itemize}
    \item $\psi^0_{w}(T)$ represents empty irregularising walks of $T$. This makes sense
    only for the parameter $d= d_T(r)+w$, so we drop it in the notation. More precisely, this parameter
    is $0$ when $T$ is locally irregular under the assumption that the degree of $r$ is replaced by $d_T(r)+w$,
    and $+\infty$ otherwise.
    
    \item $\psi^{\rm io}_{w,d}(T)$ (where ``${\rm io}$'' stands for ``In and Out of the root'') is the length of a shortest $w$-irregularising walk of $T$ starting and ending at $r$, such that the degree of $r$ in $T+W$ is $d-w$.
    
    \item $\psi^{\rm i}_{w,d}(T)$ is the length of a shortest $w$-irregularising walk of $T$ starting at $r$ and ending at any vertex (including $r$), such that the degree of $r$ in $T+W$ is $d-w$.
    
    \item $\psi^{\rm r}_{w,d}(T)$ is the length of a shortest $w$-irregularising walk of $T$ traversing $r$ but not necessarily starting nor ending at $r$, such that the degree of $r$ in $T+W$ is $d-w$.
    
    \item $\psi^{\not{\rm r}}_{w,d}(T)$ is the length of a shortest $w$-irregularising walk of $T$ not necessarily containing $r$, such that the degree of $r$ in $T+W$ is $d-w$.
\end{itemize}

Note that, by Proposition~\ref{proposition:walk_shapes}, in any graph $G$ of maximum degree~$\Delta$ with an edge $xy$, any irregularising walk traversing $xy$ more than $2(d(x)+d(y)-1)$ times can be simplified by removing a certain number of half-turns on $xy$ to get a shorter irregularising walk.
In particular, this implies that any vertex has degree $\mathcal{O}(\Delta^2)$ in $G+W$, for any shortest irregularising walk $W$ of $G$.
From this, for the parameters above, we only have to consider when $w$ and $d$ are $\mathcal{O}(\Delta^2)$.

It is important too to raise that rooted trees can be defined inductively as follows:
\begin{itemize}
    \item The isolated vertex is a rooted tree.
    
    \item Given two rooted trees $T'$ and $T''$ with root $r$ and $r''$, respectively, a new rooted tree $T$ with root $r$ is obtained when starting from the disjoint union of $T'$ and $T''$, and adding the edge $rr''$. We denote this operation as $T'\uparrow T''$.
    Put differently, this operation adds $T''$ as a subtree of $T'$.
\end{itemize}

\noindent Note in particular that if $T$ is a rooted tree with root $r$, then the rooted tree with root $u$ of degree~$1$ obtained from $T$ by adding a new pendant edge $ur$ is nothing but $T_u \uparrow T$, where $T_u$ is the rooted tree containing $u$ as its only vertex (and, thus, root). From this, it is easy to see that rooted trees can be constructed iteratively through the operation $\uparrow$, by essentially adding subtrees one at a time, somewhat adding one edge each step.

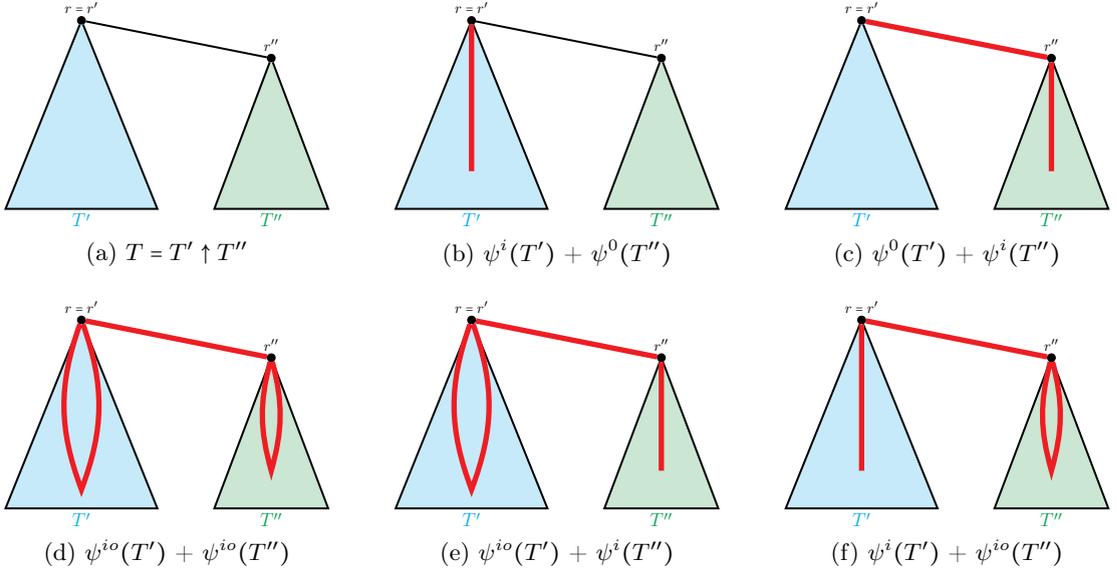
\begin{figure}[!t]
 	\centering
 	
 	\subfloat[$T=T'\uparrow T''$]{
    \scalebox{0.5}{
	\begin{tikzpicture}[inner sep=0.7mm]
	\node at (0,-5.3){\Large\textcolor{Cyan}{$T'$}};
	\node at (5,-5.3){\Large\textcolor{Green}{$T''$}};
	
	\node[draw,circle,line width=1pt,fill=black](r) at (0,0)[label=above:{$r=r'$}]{};
	\draw[line width=1.5pt,draw,black,fill=Cyan!20] (r) -- (-2,-5) -- (2,-5) -- (r);
	
	\node[draw,circle,line width=1pt,fill=black](r2) at (5,-1)[label=above:{$r''$}]{};
	\draw[line width=1.5pt,draw,black,fill=Green!20] (r2) -- (3.5,-5) -- (6.5,-5) -- (r2);
	
	\draw[line width=1.5pt,draw,black] (r) -- (r2);
	\end{tikzpicture}
    }
    }
    \hspace{10pt}
    \subfloat[$\psi^i(T')$ + $\psi^0(T'')$]{
    \scalebox{0.5}{
	\begin{tikzpicture}[inner sep=0.7mm]
	\node at (0,-5.3){\Large\textcolor{Cyan}{$T'$}};
	\node at (5,-5.3){\Large\textcolor{Green}{$T''$}};
	
	\node[draw,circle,line width=1pt,fill=black](r) at (0,0)[label=above:{$r=r'$}]{};
	\draw[line width=1.5pt,draw,black,fill=Cyan!20] (r) -- (-2,-5) -- (2,-5) -- (r);
	
	\node[draw,circle,line width=1pt,fill=black](r2) at (5,-1)[label=above:{$r''$}]{};
	\draw[line width=1.5pt,draw,black,fill=Green!20] (r2) -- (3.5,-5) -- (6.5,-5) -- (r2);
	
	\draw[line width=1.5pt,draw,black] (r) -- (r2);
	
	\draw [-,line width=4pt,color=Red] (r) -- (0,-4);
	\end{tikzpicture}
    }
    }
    \hspace{10pt}
    \subfloat[$\psi^{0}(T')$ + $\psi^{i}(T'')$]{
    \scalebox{0.5}{
	\begin{tikzpicture}[inner sep=0.7mm]
	\node at (0,-5.3){\Large\textcolor{Cyan}{$T'$}};
	\node at (5,-5.3){\Large\textcolor{Green}{$T''$}};
	
	\node[draw,circle,line width=1pt,fill=black](r) at (0,0)[label=above:{$r=r'$}]{};
	\draw[line width=1.5pt,draw,black,fill=Cyan!20] (r) -- (-2,-5) -- (2,-5) -- (r);
	
	\node[draw,circle,line width=1pt,fill=black](r2) at (5,-1)[label=above:{$r''$}]{};
	\draw[line width=1.5pt,draw,black,fill=Green!20] (r2) -- (3.5,-5) -- (6.5,-5) -- (r2);
	
	\draw[line width=1.5pt,draw,black] (r) -- (r2);
	
	\draw[line width=4pt,draw,Red] (r) -- (r2);
	\draw [-,line width=4pt,color=Red] (r2) -- (5,-4);
	\end{tikzpicture}
    }
    }

    \subfloat[$\psi^{io}(T')$ + $\psi^{io}(T'')$]{
    \scalebox{0.5}{
	\begin{tikzpicture}[inner sep=0.7mm]
	\node at (0,-5.3){\Large\textcolor{Cyan}{$T'$}};
	\node at (5,-5.3){\Large\textcolor{Green}{$T''$}};
	
	\node[draw,circle,line width=1pt,fill=black](r) at (0,0)[label=above:{$r=r'$}]{};
	\draw[line width=1.5pt,draw,black,fill=Cyan!20] (r) -- (-2,-5) -- (2,-5) -- (r);
	
	\node[draw,circle,line width=1pt,fill=black](r2) at (5,-1)[label=above:{$r''$}]{};
	\draw[line width=1.5pt,draw,black,fill=Green!20] (r2) -- (3.5,-5) -- (6.5,-5) -- (r2);
	
	\draw[line width=1.5pt,draw,black] (r) -- (r2);
	
	\draw[line width=4pt,draw,Red] (r) -- (r2);
	\draw [-,line width=4pt,color=Red] (r) to[bend left=20] (0,-4.5) to[bend left=20] (r);
	\draw [-,line width=4pt,color=Red] (r2) to[bend left=15] (5,-4) to[bend left=15] (r2);
	\end{tikzpicture}
    }
    }
    \hspace{10pt}
    \subfloat[$\psi^{io}(T')$ + $\psi^{i}(T'')$]{
    \scalebox{0.5}{
	\begin{tikzpicture}[inner sep=0.7mm]
	\node at (0,-5.3){\Large\textcolor{Cyan}{$T'$}};
	\node at (5,-5.3){\Large\textcolor{Green}{$T''$}};
	
	\node[draw,circle,line width=1pt,fill=black](r) at (0,0)[label=above:{$r=r'$}]{};
	\draw[line width=1.5pt,draw,black,fill=Cyan!20] (r) -- (-2,-5) -- (2,-5) -- (r);
	
	\node[draw,circle,line width=1pt,fill=black](r2) at (5,-1)[label=above:{$r''$}]{};
	\draw[line width=1.5pt,draw,black,fill=Green!20] (r2) -- (3.5,-5) -- (6.5,-5) -- (r2);
	
	\draw[line width=1.5pt,draw,black] (r) -- (r2);
	
	\draw[line width=4pt,draw,Red] (r) -- (r2);
	\draw [-,line width=4pt,color=Red] (r) to[bend left=20] (0,-4.5) to[bend left=20] (r);
	\draw [-,line width=4pt,color=Red] (r2) -- (5,-4);
	\end{tikzpicture}
    }
    }
    \hspace{10pt}
    \subfloat[$\psi^{i}(T')$ + $\psi^{io}(T'')$]{
    \scalebox{0.5}{
	\begin{tikzpicture}[inner sep=0.7mm]
	\node at (0,-5.3){\Large\textcolor{Cyan}{$T'$}};
	\node at (5,-5.3){\Large\textcolor{Green}{$T''$}};
	
	\node[draw,circle,line width=1pt,fill=black](r) at (0,0)[label=above:{$r=r'$}]{};
	\draw[line width=1.5pt,draw,black,fill=Cyan!20] (r) -- (-2,-5) -- (2,-5) -- (r);
	
	\node[draw,circle,line width=1pt,fill=black](r2) at (5,-1)[label=above:{$r''$}]{};
	\draw[line width=1.5pt,draw,black,fill=Green!20] (r2) -- (3.5,-5) -- (6.5,-5) -- (r2);
	
	\draw[line width=1.5pt,draw,black] (r) -- (r2);
	
	\draw[line width=4pt,draw,Red] (r) -- (r2);
	\draw [-,line width=4pt,color=Red] (r) -- (0,-4);
	\draw [-,line width=4pt,color=Red] (r2) to[bend left=15] (5,-4) to[bend left=15] (r2);
	\end{tikzpicture}
    }
    }
    
\caption{Expressing optimal irregularising walks in trees.
(a) Any rooted tree $T$ with root $r$ can be expressed as $T=T'\uparrow T''$, for some rooted trees $T'$ and $T''$ with roots $r'=r$ and $r''$, respectively.
Through (b) to (f), we illustrate how some optimal irregularising non-closed walks of $T$ with end-vertex $r$ (corresponding to parameter $\psi^i_{w,d}(T)$)
can be expressed as particular types of optimal irregularising walks of $T'$ and $T''$.
Walks are in red; for readability, we represent their supporting edges only, and omit possible half-turns.
We insist on the fact that not all possible cases (for $\psi^i_{w,d}(T)$) are displayed here. Also, subcaptions indicate concatenations of types of paths; they are not intended to display exact expressions.
\label{figure:trees}}
\end{figure}

Theorem~\ref{theorem:trees-polynomial} can now be proved mainly using the fact that, for any rooted tree $T=T' \uparrow T''$, any parameter $\mlw_{w,d}(T)$ can be deduced from parameters $\psi^0_{w'}$, $\psi^{\rm io}_{w',d'}$, $\psi^{\rm i}_{w',d'}$, $\psi^{\rm r}_{w',d'}$, $\psi^{\not{\rm r}}_{w',d'}$ of $T'$ and $T''$ (see Figure~\ref{figure:trees} for an illustration), which can be computed by bottom-up dynamic programming, since a $w$-irregularising walk of $T$ can be seen as a concatenation of (possibly empty) $w'$-irregularising walks of $T'$ and $T''$. 
A crucial point for the polynomiality is that all of $w,d,w',d'$ are $\mathcal{O}(\Delta^2)$, where $\Delta=\Delta(T)$. So, there are a polynomial number of parameters to be computed for each vertex, each of which can be computed in polynomial time; over all vertices, such an algorithm thus runs in polynomial time.


\section{Conclusion}\label{section:ccl}

In this work, we have introduced and studied a new problem that is, essentially, a restriction of proper labellings, where locally irregular multigraphs are built upon graphs by adding the edges of a walk. As described in Section~\ref{section:def}, this problem is rather flexible, in that several of the parameters involved can be set in different ways. For a first work on the topic, we chose to focus on shortest irregularising walks, and on studying the parameter $\mlw$. In Section~\ref{section:early-properties}, we exhibited several early properties of this parameter, before providing general upper bounds on $\mlw$ in Section~\ref{section:bounds}, through Corollaries~\ref{corollary:bound-general} and~\ref{corollary:bound-chromatic}. We then focused more thoroughly on several classes of graphs in Section~\ref{section:particular-classes}, for which we determined the exact value of $\mlw$. Last, in Section~\ref{section:algo}, we proved mainly, through Theorem~\ref{theorem:npc-walks}, that determining $\mlw(G)$ for a given graph $G$ is \np-complete. On the other hand, we proved that this can be achieved in polynomial time when $G$ is a tree.

Although our results lead to some understanding of irregularising walks, several aspects remain open and could be subject to further work on the topic. In particular, we believe the following directions could be appealing, and worth studying further.

\begin{itemize}
    \item As illustrated by Corollary~\ref{corollary:bound-general}, we believe expressing bounds on the parameter $\mlw$ in terms of the graph's size is appropriate. As proved, for every nice graph $G$ of size $m$ we have $\mlw(G) \leq 4m$, and one can naturally wonder how tight this upper bound is, in general. Through \textit{e.g.}~Theorem~\ref{theorem:paths} and Corollary~\ref{corollary:cycles}, we are aware of graphs $G$ of size $m$ for which $\mlw(G)$ is about $2m$. 
    
    It seems to us that one can also come up with graphs $G$ of size $m$ for which $\mlw(G)$ is about $3m$. For instance, consider any subdivided star $G$ where the center has degree~$k$ and all branches have length~$n$. Assuming $k$ and $n$ are large, note that any irregularising walk $W$ of $G$ essentially has to start from a branch and reach the center, then go through most vertices of $k-2$ other branches and come back to the root, and then finish in the last branch. Given how irregularising walks of paths behave (recall Theorem~\ref{theorem:paths}), the fact that the problem falls down to irregularising closed walks of paths going through one end-vertex, yields that $\|W\|$ should tend to $3m$ as $n$ and $k$ grow large.
    
    From a more general perspective, we wonder whether there are worse graphs (when comparing $\mlw$ and the size parameter). Our guess is that, perhaps, nice graphs of size $m$ admit irregularising walks of size about $3m$.

    \begin{conjecture}\label{conjecture:3m}
    For every nice graph $G$ of size $m$, we have $\mlw(G) \leq 3m$.
    \end{conjecture}

    Perhaps the following arguments provide some additional support to Conjecture~\ref{conjecture:3m}. If we consider any graph $G$ of size $m$ and an irregularising walk $W$ of $G$ obtained through Theorem~\ref{theorem:bound-general} and meeting the properties of Proposition~\ref{proposition:walk_shapes}, then we can essentially count at most one conflict per edge, and, half-turns apart, each traversed edge is traversed once or twice; in the worst-case scenario, this yields that the length of $W$ lies in-between $3m$ and $4m$. So, we could perhaps get closer to the conjectured $3m$ bound by investigating the size of $E_e$, in the proof of Proposition~\ref{proposition:walk_shapes}.
    
    \item In Section~\ref{section:particular-classes}, we deduced the exact value of $\mlw(G)$ when $G$ belongs to common graph classes. We believe it would be interesting to investigate other usual graph classes as well. Among interesting candidate classes, let us mention cographs (which makes sense, since we considered complete graphs and complete bipartite graphs), planar graphs (for which, as mentioned after the proof of Theorem~\ref{theorem:npc-walks}, computing $\mlw$ remains \np-complete; while we deduce some upper bound on $\mlw$ through Corollary~\ref{corollary:bound-chromatic} since planar graphs have chromatic number at most~$4$ by the Four-Colour Theorem), graphs of low maximum degree (to which, again, Theorem~\ref{theorem:npc-walks} applies, and similarly for Corollary~\ref{corollary:bound-chromatic} through Brooks' Theorem), and particular classes of trees such as complete balanced $k$-ary trees (for which we know $\mlw$ can be determined in polynomial time, recall Theorem~\ref{theorem:trees-polynomial}). 
    Maybe it could be interesting too to investigate graphs with bounded treedepth: these graphs admit spanning trees with bounded depth, thus with many leaves, which, as described right after Corollary~\ref{corollary:bound-general}, is a very favourable case for the approach we developed in the proof of Theorem~\ref{theorem:bound-general}.
    
    \item Regarding algorithmic aspects and our results from Section~\ref{section:algo}, many interesting questions remain open. For instance, regarding Theorem~\ref{theorem:npc-walks}, recall that our \np-completeness result holds for graphs of maximum degree $7$, and we wonder whether the same also holds for graphs with smaller maximum degree. 
    Given that \textsc{Irregularising Walk} is \np-complete, one can also wonder about parameterised complexity. For instance, through Courcelle's Theorem, one can establish that \textsc{Irregularising Walk} is \textsf{FPT} when parameterised by the maximum degree and the treewidth. Note that a similar result does not hold when parameterising by the maximum degree only, since \textsc{Irregularising Walk} is \np-complete for graphs with bounded maximum degree (Theorem~\ref{theorem:npc-walks}). On the other hand, we believe that our approach for proving Theorem~\ref{theorem:trees-polynomial} could generalise to graphs of bounded treewidth, to show that \textsc{Irregularising Walk} is \textsf{FPT} when parameterised by the treewidth only. Such a generalisation could build upon the following elements. Given a graph $G$ with a tree-decomposition where bags have size~$k$, every subgraph $H$ given by the decomposition would now have a set of $k$ roots $r_1,\dots,r_k$, and we could express how an optimal irregularising walk $W$ of $G$ behaves within $H$, in terms of particular optimal irregularising walks through $r_1,\dots,r_k$. In particular, one has to take into account that $W$ can ``enter'' and ``leave'' $H$ multiple times through some of $r_1,\dots,r_k$. However, each of these bits of subwalks of $W$ is still of bounded length. Thus, there is a way to express all ways for $W$ to be shaped w.r.t.~$H$, that involves a polynomial (in $k$) number of length parameters. 
    
    The fact that \textsc{Irregularising Walk} is \textsf{FPT} when parameterised by the treewidth would yield a significant contrast to the fact that \textsc{Irregularising Walk} is \textsf{W[1]}-hard when parameterised by the clique-width (recall our arguments right after the proof of Theorem~\ref{theorem:npc-walks}).
    This is another example of the clique-width's price of generality~\cite{FGLS09}.

    Of course, we also wonder about other interesting combinations of parameters. For instance, \textsc{Irregularising Walk} is \textsf{FPT} when parameterised by the maximum degree $\Delta$ and the length $k$ of the shortest irregularising walks. To see this is true, just note that any graph $G$ either is locally irregular (and $\mlw(G)=0$), or there is an edge $uv$ of $G$ with $d(u)=d(v)$, in which case, since any irregularising walk of $G$ has to go through $u$ or $v$, we fall down to finding irregularising walks (of length at most $k$) in a subgraph of order $\mathcal{O}\left(\Delta^k\right)$ (in which it then suffices to enumerate all possible walks of length at most~$k$). Recall that we also mentioned that \textsc{Irregularising Walk} is \textsf{XP} when parameterised by~$k$, and we leave open the question of whether it is \textsf{FPT} for this parameter.
    
    \item Recall that we also introduced the parameters $\mew$ and $\mvw$. By results provided in Section~\ref{section:early-properties}, there are relationships between these parameters and $\mlw$; thus, most results on $\mlw$ provided in this work also yield something for $\mew$ and $\mvw$. In our opinion these two other parameters are also natural to consider, so they could deserve further attention towards a better understanding. Note that some of our results (such as Corollary~\ref{cor:np-c}) already make a step in that direction.
\end{itemize}

\end{document}